\documentclass[a4paper,11pt]{article}
\pdfoutput=1

\usepackage[a4paper,top=3.5cm,right=3.1cm,left=3.1cm,bottom=3.5cm]{geometry}
\usepackage{amsthm}
\usepackage{amssymb}
\usepackage{amsmath}
\usepackage{bm}
\usepackage{cite}
\usepackage{graphicx}

\usepackage{hyperref}

\newcommand\bra[1]{{\langle#1|}}
\newcommand\ket[1]{{|#1\rangle}}

\newcommand\bbra[1]{{\langle\!\langle#1|}}
\newcommand\kket[1]{{|#1\rangle\!\rangle}}

\newcommand{\llangle}{\langle\!\langle}
\newcommand{\rrangle}{\rangle\!\rangle}

\newcommand\id{\mathbf{1}}

\numberwithin{equation}{section}
\newtheorem{definition}{Definition}[section]
\newtheorem{example}[definition]{Example}
\newtheorem{conjecture}[definition]{Conjecture}
\newtheorem{proposition}[definition]{Proposition}
\newtheorem{theorem}[definition]{Theorem}
\newtheorem{corollary}[definition]{Corollary}
\newtheorem{lemma}[definition]{Lemma}

\begin{document}

\setcounter{page}{0}
\pagestyle{empty}
\null

\vfill

\begin{center}
    {\LARGE  {\sffamily
    Matrix product construction for Koornwinder polynomials and fluctuations of the current in the open ASEP
    }}
    \vspace{60pt}

    \begin{large}
     C. Finn$^{a}$\footnote{caley.finn@lapth.cnrs.fr},
     and M. Vanicat$^{a}$\footnote{matthieu.vanicat@lapth.cnrs.fr}
     \\[.41cm] 
    $^{a}$  LAPTh,
     CNRS - Universit{\'e} Savoie Mont Blanc\\
       9 chemin de Bellevue, BP 110, F-74941  Annecy-le-Vieux Cedex, 
    France. 
    \end{large}
\end{center}

\vfill

\begin{abstract}
Starting from the deformed current-counting transition matrix for the open boundary ASEP, we prove that with a
further deformation, the symmetric Koornwinder polynomials for partitions with equal row lengths appear as the
normalisation of the twice deformed ground state.  We give a matrix product construction for this ground state
and the corresponding symmetric Koornwinder polynomials.  Based on the form of this construction and numerical
evidence, we conjecture a relation between the generating function of the cumulants of the current, and a certain
limit of the symmetric Koornwinder polynomials.
\end{abstract}
\vfill
\vfill
\begin{flushright}
October 2016\\
LAPTH-061/16
\end{flushright}
\vfill
\newpage
\pagestyle{plain}

\setcounter{footnote}{0}

\section{Introduction}

The asymmetric simple exclusion process (ASEP) \cite{Spitzer70,Liggett85} has become over the last decades a paradigmatic model in
non-equilibrium statistical mechanics \cite{Derrida98,ChouMZ11}.  It is an example of a physical system exhibiting a
macroscopic current in a stationary regime.  Such systems, which cannot be described by the usual thermal
equilibrium formalism, can be seen as the simplest out-of-equilibrium situation one can imagine
\cite{KatzLS84,KrapivskyRB10,SchmittmannZ95}.  The large deviation function of the current is a key tool in
the study of these systems, and has been proposed as a generalization of the traditional thermodynamic potentials
to non-equilibrium systems \cite{Touchette09}.

The large deviation function is studied via the generating function of the cumulants of the current, and
several exact results have been obtained for the ASEP using a deformed current-counting transition matrix
\cite{DerridaEM95,deGierE11,LazarescuM11,GorissenLMV12,Lazarescu13jphysA,LazarescuP14}.  The works
\cite{DerridaEM95,LazarescuM11,GorissenLMV12,Lazarescu13jphysA,LazarescuP14} build upon the matrix product
method \cite{DerridaEHP93}, used to compute the stationary state of the undeformed ASEP.  In this work we
continue in this vein, relying particularly on the integrability of the  ASEP \cite{Baxter82,Sklyanin88} and
the connection between integrability and the matrix product method \cite{SasamotoW97,CrampeRV14}.  We note
also the approach of \cite{deGierE11}, in which the Bethe ansatz was used to obtain the cumulant
generating function in the thermodynamic limit.

The integrable structure of the ASEP gives rise to a connection with Hecke algebras, and Macdonald and
Koornwinder polynomials \cite{CantiniDGW15,Cantini15,CantiniGDGW16}.  The Macdonald polynomials are associated with the
periodic system, and in \cite{CantiniDGW15} this connection is exploited to derive a matrix product formula
for the symmetric Macdonald polynomials.  The connection between the open system and
Koornwinder polynomials was first identified in \cite{Cantini15}, then fully established in
\cite{CantiniGDGW16}.  However a matrix product formula, and the link to the general form of
the Koornwinder polynomials is still lacking.

The aim of this paper is to exploit the integrable structure of the ASEP with deformed current-counting matrix,
to make a connection to the general form of the Koornwinder polynomials.  This in turns leads to a connection
between the symmetric Koornwinder polynomials and the generating function of the cumulants of the current.  We
do this by introducing scattering relations and $q$KZ equations with a further deformation, through which we
define a twice deformed ground state vector.  We give a matrix product construction of this ground state
vector and of the symmetric Koornwinder polynomial associated with it.  This leads us to
conjecture a beautiful relation between the generating function of the cumulants of the current, and a
certain limit of symmetric Koornwinder polynomials.

In this work we consider the ASEP with \emph{partial} asymmetry, but it would be interesting also to consider
the \emph{totally} asymmetric simple exclusion process (TASEP)\footnote{That is with $q = 0$ in the model, as
defined below.}. The TASEP exhibits
broadly similar behaviour to the general ASEP physically, but often the involved mathematical expressions are much
simpler.  The stationary state of the TASEP can be expressed in matrix product form \cite{DerridaEHP93}, but
was also given by directly solving certain recursion relations \cite{SchuetzD93}.  The results relating to
current fluctuations in the TASEP \cite{LazarescuM11} are also much simpler than those for the general ASEP.
In our notation, the TASEP relates to the $t \to \infty$ limit of the Koornwinder polynomials, which has been
previously studied \cite{OrrS13}. Thus it would
be interesting to study this limit in the TASEP context, and to see if any simplifications occur.

In the following subsections we review briefly the main tools that are needed in this work: 
(i) the ASEP, the current-counting deformation of the associated Markov matrix, and the link with 
the generating function of cumulants of the current, (ii) the Hecke algebra and Koornwinder polynomials, (iii) the
integrable structure of the ASEP.

The outline of the rest of the paper is as follows: in section \ref{sec:twice_deformed}, we introduce
scattering relations and $q$KZ equations, whose solutions define a twice deformed ground state vector.  We
point out the connection between components of this deformed ground state vector and non-symmetric Koornwinder
polynomials.  In section \ref{sec:matrixproduct} we construct solutions of the $q$KZ equations in matrix
product form, which then allows us to make the connection to symmetric Koornwinder polynomials.  Finally in
section \ref{sec:currentKoornwinder} we conjecture that the twice deformed ground state vector converges in a
specific limit to the ground state of the deformed Markov matrix.  This conjecture allows us
to express the generating function of the cumulants of the current as a certain limit of symmetric Koornwinder
polynomials.

\subsection{The ASEP and current fluctuations}
The open boundary ASEP is a stochastic model of interacting particles
set on a one-dimensional lattice.  We consider a lattice of length $N$, where each lattice site may be
occupied by a single particle, or is empty.  In the bulk of the lattice particles hop right one site with rate
$p$, and left with rate $q$, so long as the target site is empty (the exclusion rule).  With open boundaries,
particles may enter and exit at the first and last sites.  If site $1$ is empty (occupied), a particle is
injected (extracted) with rate $\alpha$ ($\gamma$).  At site $N$, particles are extracted with rate $\beta$ and
injected with rate $\delta$.  These rules, summarised in figure~\ref{fig:aseprules}, describe a continuous time
Markov process.
\begin{figure}[h]
    \centering
    \includegraphics{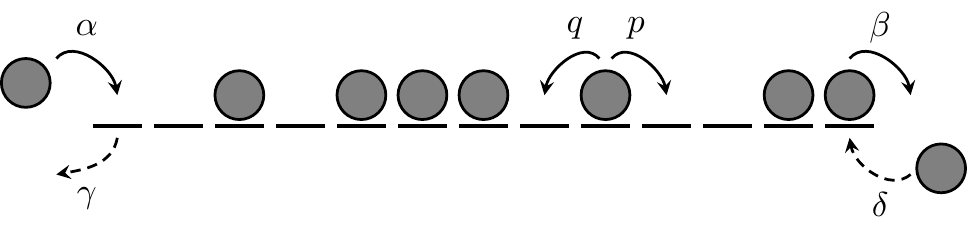}
    \caption{Transition rates for the ASEP with open boundaries.}
    \label{fig:aseprules}
\end{figure}

A Markov process is defined by its transition matrix, and to give this matrix we must first specify a basis.
To each site $i$ we attach a boolean variable $\tau_i \in \{0,1\}$ indicating if the site is empty ($\tau_i =
0$) or occupied ($\tau_i = 1$).  The state of a single site is represented by a vector $\ket{\tau_i} \in
\mathbb{C}^2$, where
\begin{equation*}
    \ket{0} = \begin{pmatrix}
        1
        \\
        0
    \end{pmatrix},
    \qquad
    \ket{1} = \begin{pmatrix}
        0
        \\
        1
    \end{pmatrix}.
\end{equation*}
The state of the lattice is given by a vector
$\ket{\bm\tau} = \ket{\tau_1, \ldots, \tau_N} \in \left(\mathbb{C}^2\right)^{\otimes N}$ with
\begin{equation*}
    \ket{\tau_1, \ldots \tau_N} = \ket{\tau_1} \otimes \ldots \otimes \ket{\tau_N}.
\end{equation*}
The ASEP transition rates are then encoded in the transition matrix $M(\xi = 1)$, where\footnote{The unusual
normalisation is to ease the notation in later sections.}
\begin{equation}
    M(\xi)
    =
    \sqrt{\alpha \gamma} B_1(\xi) + \sum_{i=1}^{N-1} \sqrt{p q} w_{i,i+1} + \sqrt{\beta \delta} \overline{B}_N,
\end{equation}
and
\begin{equation}
    \sqrt{\alpha \gamma} B(\xi) = \begin{pmatrix}
        -\alpha & \xi^{-1} \gamma \\
        \xi \alpha & -\gamma
    \end{pmatrix},
    \qquad
    \sqrt{\beta \delta} \overline{B} = \begin{pmatrix}
        -\delta & \beta \\
        \delta & -\beta
    \end{pmatrix},
    \qquad
    \sqrt{p q} w = \begin{pmatrix}
        0 & 0 & 0 & 0
        \\
        0 & -q & p & 0
        \\
        0 & q & -p & 0
        \\
        0 & 0 & 0 & 0
    \end{pmatrix}.
\end{equation}
The indices on the matrices indicate the sites on which they act.  The matrix $M(\xi)$ is stochastic only
for $\xi = 1$, but the introduction of this fugacity allows the study of the current generating function, as will be
discussed below.  The stochastic matrix $M(\xi = 1)$ has a unique eigenvector with eigenvalue $0$, that is
\begin{equation*}
    M(1) \ket{\Psi} = 0,
    \qquad
    \ket{\Psi} = \sum_{\bm\tau} \psi_{\bm\tau} \ket{\bm\tau}.
\end{equation*}
Normalising this vector gives the stationary distribution of the system: letting
\begin{equation*}
    \mathcal{Z} = \langle 1 | \Psi \rangle,
    \qquad
    \bra{1} = (1, 1)^{\otimes N},
\end{equation*}
the stationary probability of a configuration $\bm\tau$ is
\begin{equation*}
    P_{\text{stat}}(\bm{\tau}) = \frac{1}{\mathcal{Z}} \psi_{\bm\tau}.
\end{equation*}

If we now consider the deformed transition matrix $M(\xi)$, then the deformed ground state vector
satisfies
\begin{equation*}
    M(\xi) \ket{\Psi(\xi)} = \Lambda_0(\xi) \ket{\Psi(\xi)},
\end{equation*}
with $\Lambda_0(\xi) \to 0$ as $\xi \to 1$. The eigenvalue $\Lambda_0(\xi)$ for general $\xi$ is an object of
prime interest in the context of out-of-equilibrium statistical physics, because of its connection to the
generating function of the cumulants of the current, $E(\mu)=\Lambda_0(e^{\mu})$.    It has been shown recently
\cite{LazarescuM11,GorissenLMV12,Lazarescu13jphysA,LazarescuP14} that the cumulants of the current for finite systems can be extracted
analytically at any order at the price of solving non-linear implicit equations.
  The Legendre transformation
of $E(\mu)$ provides the large deviation function of the particle current in the stationary state,
\begin{equation*}
G(j)=\min\limits_{\mu}\big(\mu j-E(\mu)\big),
\end{equation*}
which is expected to be a possible generalisation of thermodynamic potential to non-equilibrium systems
\cite{Touchette09}.
In words, $G(j)$ describes the non-typical fluctuations of the
mean particle flux.
More precisely if we denote by $Q_T$ the algebraic number of particles exchanged between
the system and the left reservoir during the time interval $[0,T]$, then $G(j)$ is characterised by
$P(Q_T/T=j)\sim \exp(-TG(j))$ for large $T$.
   The reader can refer, for instance, to \cite{GorissenLMV12} for more details.

The eigenvalue $\Lambda_0(\xi)$ is invariant under the Gallavotti--Cohen symmetry \cite{GallavottiC95,LebowitzS99}
\begin{equation}
    \xi \to \xi' = \frac{\gamma \delta}{\alpha \beta} \left(\frac{q}{p}\right)^{N-1} \xi^{-1}.
    \label{eq:GCsym}
\end{equation}
This translates immediately into a symmetry on the large deviation function of the particle current.
\begin{equation}
 G(j)-G(-j)=j\ln \left(\frac{\gamma \delta}{\alpha \beta} \left(\frac{q}{p}\right)^{N-1}\right).
\end{equation}
This symmetry arises from the relation between the transition matrix and its transpose:
\begin{equation}
    M(\xi') = U_{\text{GC}} M(\xi)^T U_{\text{GC}}^{-1},
    \label{eq:GCTranspose}
\end{equation}
where $\xi'$ is as defined in \eqref{eq:GCsym} and
\begin{equation}
    U_{\text{GC}}
    =
    \begin{pmatrix}
        1 & 0
        \\
        0 & \frac{\delta}{\beta}\left(\frac{q}{p}\right)^{N-1}
    \end{pmatrix}
    \otimes
    \begin{pmatrix}
        1 & 0
        \\
        0 & \frac{\delta}{\beta}\left(\frac{q}{p}\right)^{N-2}
    \end{pmatrix}
    \otimes \dots \otimes
    \begin{pmatrix}
        1 & 0
        \\
        0 & \frac{\delta}{\beta}
    \end{pmatrix}.
    \label{eq:UGC}
\end{equation}
The relation \eqref{eq:GCTranspose} implies that $M(\xi)$ and $M(\xi')$ have the same spectrum and thus the largest
eigenvalue is the same:
\begin{equation*}
    \Lambda_0(\xi') = \Lambda_0(\xi).
\end{equation*}
As a further consequence of this symmetry, given a solution of the left eigenvalue problem
\begin{equation*}
    \bra{\Phi({\xi})} M(\xi) = \Lambda(\xi) \bra{\Phi(\xi)}.
\end{equation*}
there is a corresponding solution of the right eigenvalue problem
\begin{equation*}
    M(\xi') \ket{\Psi({\xi'})} = \Lambda(\xi) \ket{\Psi(\xi')},
\end{equation*}
with
\begin{equation*}
    \ket{\Psi(\xi')} = U_{\text{GC}} \ket{\Phi(\xi)},
\end{equation*}
and $\xi'$ as defined in \eqref{eq:GCsym}.
Note that here and in the following we use the convention $\bra{\cdot}^T=\ket{\cdot}$, where $^T$ denotes the usual transposition.
We will explain in the following sections the connection that can be made between the ground state $\ket{\Psi(\xi)}$ and the 
theory of Koornwinder polynomials that we present now.

\subsection{Koornwinder polynomials}
We now introduce the symmetric and non-symmetric Koornwinder polynomials, which form the other main theme of
this work.  The symmetric Koornwinder polynomials \cite{Koornwinder92, vanDiejen96} are a family of
multivariate orthogonal polynomials generalising the Askey--Wilson polynomials.  The symmetric
Koornwinder polynomials can be constructed from their non-symmetric counterparts, which arise from the
polynomial representation of the affine Hecke algebra of type $C_N$ \cite{Sahi99, Stokman00}.

\subsubsection{Hecke algebra}
The affine Hecke algebra of type $C_N$ is generated by elements $T_0, T_1, \ldots, T_N$, with parameters
$t^{1/2}$, $t_0^{1/2}$ and $t_N^{1/2}$.  The generators satisfy the quadratic relations,
\begin{align*}
    & \left(T_0 - t_0^{1/2}\right)\left(T_0 + t_0^{-1/2}\right) = 0,
    \\
    & \left(T_i - t^{1/2}\right)\left(T_i + t^{-1/2}\right) = 0, \qquad 1 \le i \le N-1,
    \\
    & \left(T_N - t_N^{1/2}\right)\left(T_N + t_N^{-1/2}\right) = 0,
\end{align*}
the braid relations
\begin{align*}
    T_1 T_0 T_1 T_0 & = T_0 T_1 T_0 T_1
    \\
    T_i T_{i+1} T_i & = T_{i+1} T_i T_{i+1}, \qquad 1 \le i \le N - 2,
    \\
    T_N T_{N-1} T_N T_{N-1} & = T_{N-1} T_N T_{N-1} T_N,
\end{align*}
and otherwise commute.  That is
\begin{equation*}
    T_i T_j = T_j T_i, \qquad |i - j| \ge 2.
\end{equation*}
The algebra contains a family of mutually commuting elements \cite{Sahi99,Lusztig89}
\begin{equation}
    Y_i = T_i \ldots T_{N-1} T_N \ldots T_0 T_1^{-1} \ldots T_{i-1}^{-1},
    \qquad
    1 \le i \le N.
    \label{eq:Yi}
\end{equation}

We are interested in the representation of this algebra due to Noumi \cite{Noumi95} (see also \cite{Sahi99}),
acting on Laurent polynomials in $x_1, \ldots, x_N$.  The Noumi representation contains three additional
parameters, $u_0^{1/2}$, $u_N^{1/2}$, and $s^{1/2}$, and is defined in terms of operators $s_i$ acting on the $x_i$
as
\begin{equation}
    s_0: x_1 \to s x_1^{-1}, \qquad s_N: x_N \to x_N^{-1},
        \qquad s_i: x_i \leftrightarrow x_{i+1}, \quad 1 \le i \le N - 1.
\label{eq:sidef}
\end{equation}
The elements $s_0, s_1, \ldots s_N$ generate the affine Weyl $W$ group of type $C_N$.  The \emph{finite} Weyl
group $W_0$ is the subgroup generated by $s_1, \ldots, s_N$.

Then in the Noumi representation, the generators of the affine Hecke algebra are given by
\begin{equation}
\begin{aligned}
    & T_0^{\pm 1} = t_0^{\pm 1/2} - t_0^{-1/2}\frac{(x_1 - a)(x_1 - b)}{x_1(x_1 - s x_1^{-1})}(1 - s_0),
        \\
    & T_i^{\pm 1} = t^{\pm 1/2} - \frac{t^{1/2} x_i - t^{-1/2} x_{i+1}}{(x_i - x_{i+1})}(1 - s_i),
       \qquad 1 \le i \le N - 1,
        \\
    & T_N^{\pm 1} = t_N^{\pm 1/2} + t_N^{-1/2}\frac{(c x_N - 1)(d x_N - 1)}{x_N(x_N - x_N^{-1})}(1 - s_N),
\end{aligned}
\label{eq:Tireps}
\end{equation}
with
\begin{equation*}
    a = s^{1/2} t_0^{1/2}u_0^{1/2}, \quad b = -s^{1/2} t_0^{1/2} u_0^{-1/2},
    \quad
    c = t_N^{1/2}u_N^{1/2}, \quad d = -t_N^{1/2} u_N^{-1/2}.
\end{equation*}
One can check directly that the definitions \eqref{eq:Tireps} satisfy the relations of the Hecke algebra.
Formally, we define the field $\mathbb{F} = \mathbb{C}(s^{1/2}, t^{1/2}, t_0^{1/2}, u_0^{1/2}, t_N^{1/2}, u_N^{1/2})$, and
let $\mathcal{R} = \mathbb{F}[x_1, \ldots, x_N]$ be the ring of Laurent polynomials in $N$ variables over
$\mathbb{F}$.  The map sending the generators of the Hecke algebra to the operators defined in
\eqref{eq:Tireps} gives a representation of the algebra on $\mathcal{R}$ \cite{Sahi99}.

Later we will see that to relate the ASEP to the Noumi representation of the Hecke algebra we should take
\begin{equation*}
    t^{1/2} = \sqrt{\frac{p}{q}},
    \qquad
    t_0^{1/2} = \sqrt{\frac{\alpha}{\gamma}},
    \qquad
    t_N^{1/2} = \sqrt{\frac{\beta}{\delta}},
\end{equation*}
and
\begin{equation*}
    u_0^{1/2} - u_0^{-1/2} = \frac{p - q + \gamma - \alpha}{\sqrt{\alpha \gamma}},
    \qquad
    u_N^{1/2} - u_N^{-1/2} = \frac{p - q + \delta - \beta}{\sqrt{\beta \delta}}.
\end{equation*}
For the remainder of this paper we will use this parameterisation in preference to the physical parameters of
the ASEP, or the combinations $a$, $b$, $c$, $d$ appearing in \eqref{eq:Tireps}.

\subsubsection{Non-symmetric Koornwinder polynomials}
Before defining the non-symmetric Koornwinder polynomials, we will introduce some notation and definitions
concerning integer vectors, $\lambda \in \mathbb{Z}^N$, with
\begin{equation*}
    \lambda = (\lambda_1, \ldots, \lambda_N).
\end{equation*}
We call such a vector a \emph{composition}.  For a given composition, $\lambda$, we write monomials
\begin{equation*}
    \mathbf{x}^\lambda = x_1^{\lambda_1} \ldots x_N^{\lambda_N}.
\end{equation*}
A \emph{partition} is a composition satisfying
\begin{equation*}
    \lambda_1 \ge \lambda_2 \ge \ldots \ge \lambda_N \ge 0.
\end{equation*}
We denote by $\lambda^+$ the unique partition obtained from a composition $\lambda$ by reordering and changing
signs so that the entries are non-negative and in decreasing order.

There are two partial orderings on compositions that will be relevant \cite{Stokman00}.  First define the
\textit{dominance order}:  for $\mu, \lambda \in \mathbb{Z}^N$
\begin{equation*}
    \mu \le \lambda \iff \sum_{i=1}^j(\mu_i - \lambda_i) \le 0, \qquad \forall j, \, 1 \le j \le N.
\end{equation*}
Then $\mu < \lambda$ if $\mu \le \lambda$ and $\mu \ne \lambda$.
The second partial ordering `$\preceq$` is defined as
\begin{equation*}
    \mu \preceq \lambda \iff \mu^+ < \lambda^+ \text{ or }
        \left(\mu^+ = \lambda^+ \text{ and } \mu \le \lambda\right).
\end{equation*}
Then, $\mu \prec \lambda$ if $\mu \preceq \lambda$ and $\mu \ne \lambda$.

\begin{definition}
The non-symmetric Koornwinder polynomial $E_\lambda(\mathbf{x})$, indexed by composition $\lambda$, is the
unique Laurent polynomial satisfying
\begin{equation*}
\begin{aligned}
    Y_i E_\lambda(\mathbf{x}) & = y(\lambda)_i E_\lambda(\mathbf{x}), \quad 1 \le i \le N, \\
    E_\lambda(\mathbf{x})
    & =
    \mathbf{x}^\lambda + \sum_{\lambda' \prec \lambda} c_{\lambda \lambda'} \mathbf{x}^{\lambda'},
\end{aligned}
\end{equation*}
where $Y_i$ is defined in \eqref{eq:Yi}, \eqref{eq:Tireps}, $y(\lambda)_i$ is the eigenvalue, and $c_{\lambda
\lambda'}$ are coefficients.
\end{definition}
The composition $\lambda$ determines the eigenvalue $y(\lambda)_i$ \cite{Sahi99}.  The
two following cases will appear directly in this work:
\begin{itemize}
    \item
For $m > 0$, $\lambda = \left((-m)^N\right) = (-m, \ldots, -m)$,
\begin{equation}
    y(\lambda)_i = t_0^{-1/2} t_N^{-1/2} s^{-m} t^{-(i - 1)}.
    \label{eq:ynegm}
\end{equation}

    \item
For $m \ge 0$, $\lambda = \left(m^N\right) = (m, \ldots, m)$,
\begin{equation}
    y(\lambda)_i = t_0^{1/2} t_N^{1/2} s^{m} t^{N - i}.
    \label{eq:yposm}
\end{equation}
\end{itemize}
However, other non-symmetric Koornwinder polynomials will appear implicitly, and we define the following
space:
\begin{definition}
For a partition $\lambda$ of length $N$, define $\mathcal{R}^\lambda$ as the space spanned by
$\{E_\mu | \mu \in \mathbb{Z}^N, \mu^+ = \lambda\}$.
\end{definition}

\subsubsection{Symmetric Koornwinder polynomials}

The symmetric Koornwinder polynomials were introduced in \cite{Koornwinder92},
 as eigenfunctions of the $s$-difference operator
\begin{equation} \label{eq:operatorD}
 D=\sum_{i=1}^{N} g_i(\mathbf{x})(T_{s,i}-1)+\sum_{i=1}^{N} g_i(\mathbf{x}^{-1})(T_{s,i}^{-1}-1),
\end{equation}
where $g_i(\mathbf{x})$ is defined by
\begin{equation} \label{eq:function_g_i}
 g_i(\mathbf{x})= \frac{(1-ax_i)(1-bx_i)(1-cx_i)(1-dx_i)}{(1-x_i^2)(1-sx_i^2)}
 \prod_{j=1\\j\neq i}^{N} \frac{(1-tx_ix_j^{-1})(1-tx_ix_j)}{(1-x_ix_j^{-1})(1-x_ix_j)},
\end{equation}
and $T_{s,i}$ is the $i^{th}$ $s$-shift operator
\begin{equation}
 T_{s,i}f(x_1,\dots,x_i,\dots,x_N)=f(x_1,\dots,sx_i,\dots,x_N).
\end{equation}

\begin{definition}
For a partition $\lambda$, the symmetric Koornwinder polynomial $P_\lambda(\mathbf{x})$ is characterised by
the eigenvalue equation
\begin{equation} \label{eq:eigenvalue_eq_symmetric_Koornwinder}
 D P_\lambda = d_\lambda P_\lambda,
\end{equation}
with eigenvalue
\begin{equation}
 d_\lambda = \sum_{i=1}^{N} \left[t_0 t_N t^{2N-i-1}(s^{\lambda_i}-1)+t^{i-1}(s^{-\lambda_i}-1)\right],
\end{equation}
and where the coefficient of $\mathbf{x}^\lambda$ in $P_\lambda$ is equal to $1$.
\end{definition}

The symmetric Koornwinder polynomials are $W_0$-invariant (that is, invariant under the action of $s_1,
\ldots, s_N$, defined in \eqref{eq:sidef}), and their relation to the non-symmetric Koornwinder polynomials
was given in \cite{Sahi99}.
\begin{theorem}[Corollary 6.5 of \cite{Sahi99}]
The symmetric Koornwinder polynomial $P_\lambda$ can be characterised as the unique $W_0$-invariant polynomial in
$\mathcal{R}^\lambda$ which has the coefficient of $\mathbf{x}^\lambda$ equal to 1.
\label{theorem:SahiSymmetric}
\end{theorem}

In order to make contact between Koornwinder polynomials and the ground state of the current-counting
deformation of the open ASEP, we need to exploit the integrable structure of the underlying the physical model,
which we now review.

\subsection{Structure of integrability}
The ASEP is an integrable model -- the deformed transition matrix $M(\xi)$ belongs to an infinite family of
commuting matrices \cite{Sklyanin88} (see also \cite{CrampeRV14,CantiniGDGW16}). 
The generating function of these commuting matrices is called the transfer matrix. The key ingredients to construct this transfer matrix are
matrices $\check{R}(x)$, $K(x)$, and $\overline{K}(x)$ satisfying the Yang--Baxter relation
\begin{equation*}
    \check{R}_i(x_2/x_3) \check{R}_{i+1}(x_1/x_3) \check{R}_i(x_1/x_2)
    =
    \check{R}_{i+1}(x_1/x_2) \check{R}_i(x_1/x_3) \check{R}_{i+1}(x_2/x_3),
\end{equation*}
left and right reflection relations
\begin{align*}
    \check{R}_1(x_2/x_1) K_1(x_2) \check{R}_1(x_1 x_2) K_1(x_1)
    & =
    K_1(x_1) \check{R}_1(x_1 x_2) K_1(x_2) \check{R}_1(x_2/x_1),
    \\
    \check{R}_1(x_1/x_2) \overline{K}_2(1/x_1) \check{R}_1(x_1 x_2) \overline{K}_2(1/x_2)
    & =
    K_2(1/x_2) \check{R}_1(x_1 x_2) K_2(1/x_1) \check{R}_1(x_1/x_2),
\end{align*}
and the unitarity conditions
\begin{equation*}
    \check{R}_i(x) \check{R}_i\left(x^{-1}\right) = 1,
    \qquad
    K(x) K\left(x^{-1}\right) = 1,
    \qquad
    \overline{K}(x) \overline{K}\left(x^{-1}\right) = 1.
\end{equation*}

For the open boundary ASEP the matrices are given by
\begin{align}
    \check{R}(x) & = 1 + r(x) w,
    \label{eq:Rcheck}
    \\
    K(x; \xi) & = 1 + k(x; t_0^{1/2},u_0^{1/2}) B(\xi),
    \label{eq:K}
    \\
    \overline{K}(x) & = 1 + k(x^{-1}; t_N^{1/2}, u_N^{1/2}) \overline{B},
    \label{eq:Kbar}
\end{align}
where
\begin{equation}
    r(x) = \frac{x - 1}{t^{-1/2} x - t^{1/2}},
    \qquad
    k(x; t_i^{1/2}, u_i^{1/2}) = \frac{x^2 - 1}{t_i^{-1/2} x^2 - (u_i^{1/2} - u_i^{-1/2}) x - t_i^{1/2}},
\end{equation}
and $w$, $B(\xi)$, $\overline{B}$ are written in terms of the Hecke parameters as
\begin{equation}
B(\xi) = \begin{pmatrix}
    -t_0^{1/2} & \xi^{-1} t_0^{-1/2} \\
    \xi t_0^{1/2} & -t_0^{-1/2}
\end{pmatrix},
\qquad
\overline{B} = \begin{pmatrix}
    -t_N^{-1/2} & t_N^{1/2} \\
    t_N^{-1/2} & -t_N^{1/2}
\end{pmatrix},
\qquad
w = \begin{pmatrix}
        0 & 0 & 0 & 0
        \\
        0 & -t^{-1/2} & t^{1/2} & 0
        \\
        0 & t^{-1/2} & -t^{1/2} & 0
        \\
        0 & 0 & 0 & 0
\end{pmatrix}.
\label{eq:BBBarwHecke}
\end{equation}
We will write $K(x)$ for $K(x; \xi)$, except when it is necessary to distinguish between values of $\xi$.
Note that the local (physical) transition matrices are obtained from \eqref{eq:Rcheck} -- \eqref{eq:Kbar} through
\begin{equation*}
    \sqrt{p q} w = (q-p)\check{R}'(1),
    \qquad
    \sqrt{\alpha \gamma} B(\xi) = \frac{1}{2}(q - p) K'(1; \xi),
    \qquad
    \sqrt{\beta \delta} \overline{B} = -\frac{1}{2}(q - p) \overline{K}'(1),
\end{equation*}
and the Gallavotti--Cohen symmetry on the transition matrices now implies that
\begin{equation*}
\begin{aligned}
    \check{R}_i(x) & = U_{\text{GC}} \check{R}_i(x)^T U_{\text{GC}}^{-1},
    \\
    K_1(x; \xi') & = U_{\text{GC}} K_1(x; \xi)^T U_{\text{GC}}^{-1},
    \\
    \overline{K}_N(x) & = U_{\text{GC}} \overline{K}_N(x)^T U_{\text{GC}}^{-1},
\end{aligned}
\end{equation*}
with $\xi'$ and $U_{\text{GC}}$ defined in \eqref{eq:GCsym} and \eqref{eq:UGC} respectively.

\section{Twice deformed inhomogeneous ground state} \label{sec:twice_deformed}

\subsection{Scattering matrices}
Instead of the transfer matrix approach, one can define scattering matrices \cite{ZinnJustin07,Cantini15},
although the two methods are closely related.  We first define a modified left boundary matrix
\begin{equation}
    \widetilde{K}(x) = K(s^{-1/2}x),
    \label{eq:Ktilde}
\end{equation}
in order to introduce the Hecke parameter $s$.  The matrix $\widetilde{K}(x)$ satisfies
deformed unitary and reflection relations
\begin{align*}
    \widetilde{K}(s x) \widetilde{K}\left(x^{-1}\right) & = 1,
    \\
    \check{R}_1(x_2/x_1) \widetilde{K}_1(x_2) \check{R}_1(s^{-1} x_1 x_2) \widetilde{K}_1(x_1)
    & =
    \widetilde{K}_1(x_1) \check{R}_1(s^{-1} x_1 x_2) \widetilde{K}_1(x_2) \check{R}_1(x_2/x_1),
\end{align*}
and has the Gallavotti--Cohen symmetry
\begin{equation*}
    \widetilde{K}_1(x; \xi') = U_{\text{GC}} \widetilde{K}_1(x; \xi)^T U_{\text{GC}}^{-1}.
\end{equation*}

For $1 \le i \le N$, define the scattering matrices
\begin{equation}
\begin{aligned}
    \mathcal{S}_i(\mathbf{x}) =
        & \check{R}_{i-1}\left(\frac{x_{i-1}}{sx_{i}}\right) \ldots \check{R}_1\left(\frac{x_1}{sx_i}\right)
        \\
        & \cdot \widetilde{K}_1\left(\frac{1}{x_i}\right)
        \\
        & \cdot \check{R}_1\left(\frac{1}{x_i x_1}\right) \ldots \check{R}_{i-1}\left(\frac{1}{x_i x_{i-1}}\right)
          \cdot \check{R}_i\left(\frac{1}{x_i x_{i+1}}\right) \ldots \check{R}_{N-1}\left(\frac{1}{x_i x_N}\right)
        \\
        & \cdot \overline{K}_N(x_i)
        \\
        & \cdot \check{R}_{N-1}\left(\frac{x_N}{x_i}\right) \ldots \check{R}_i\left(\frac{x_{i+1}}{x_{i}}\right).
\end{aligned}
\label{eq:scatteringmatrix}
\end{equation}
Using the Yang--Baxter, reflection relations, and unitarity, we see that the scattering matrices satisfy a
deformed commutation relation
\begin{equation*}
    \mathcal{S}_i(\ldots, s x_j, \ldots) \mathcal{S}_j(x_1, \ldots, x_N)
    =
    \mathcal{S}_j(\ldots, s x_i, \ldots) \mathcal{S}_i(x_1, \ldots, x_N).
\end{equation*}
With $s = 1$, $[\mathcal{S}_i(\mathbf{x}), \mathcal{S}_j(\mathbf{x})] = 0$ for all $i$, $j$, and in fact there
is a direct relation to the transfer matrix approach \cite{Sklyanin88,CrampeRV14,CantiniGDGW16}
\footnote{Note that for $s \neq1$ there is no obvious link between the scattering matrices and the usual transfer matrix, as far as we know.}: 
\begin{equation}
 \mathcal{S}_i(\mathbf{x})=t(x_i | \mathbf{x}),
\end{equation}
where $t(z | \mathbf{x})$ is the usual transfer matrix with spectral parameter $z$ and inhomogeneity parameters $ \mathbf{x}=x_1,\dots,x_N$
(the reader may refer to \cite{Sklyanin88,CrampeRV14,CantiniGDGW16} for a precise definition).
At $s=1$, we also have
the important relations
\begin{equation} \label{eq:properties_scatteringmatrices}
    \mathcal{S}_i(\mathbf{x})|_{s = x_1 = \ldots = x_N = 1} = 1,
    \qquad
    \frac{\partial}{\partial x_i} \mathcal{S}_i(\mathbf{x})|_{s = x_1 = \ldots = x_N = 1}
    = \frac{2}{p-q} M(\xi).
\end{equation}

Considering $s$ general again, we would like to find solutions of the scattering relation
\begin{equation}
    \mathcal{S}_i(\mathbf{x}) \ket{\Psi(\ldots, x_i, \ldots)} = \ket{\Psi(\ldots, s x_i, \ldots)},
    \label{eq:scattering}
\end{equation}
where
\begin{equation}
    \ket{\Psi(\mathbf{x})} = \sum_{\bm\tau} \psi_{\bm\tau}(\mathbf{x}) \ket{\bm\tau}.
\end{equation}
Taking the derivative of \eqref{eq:scattering} with respect to $x_i$ and specialising with $x_1 = \ldots = x_N
= s = 1$, this would imply
\begin{equation}
    M(\xi) \ket{\Psi(\mathbf{1})} = 0.
    \label{eq:MxiOnPsi}
\end{equation}
For $\xi = 1$ this is the unnormalised stationary vector of the ASEP with eigenvalue 0.  For $\xi \ne 1$, the
ground state eigenvalue is non-zero, and so \eqref{eq:MxiOnPsi} should not have a solution at this point (that
is $s = 1$, $\xi \ne 1$).  However, in section~\ref{sec:currentKoornwinder} we will discuss how there could be
a solution of \eqref{eq:scattering} for $s \to 1$, $\xi \ne 1$, and how it relates to the current-counting
eigenvalue.

\subsection{$q$KZ equations}
It can be checked directly that sufficient conditions for a solution of the scattering relation
\eqref{eq:scattering} are
\begin{align}
    & \check{R}_i(x_{i+1}/x_{i}) \ket{\Psi(\ldots, x_i, x_{i+1}, \ldots)} = \ket{\Psi(\ldots, x_{i+1}, x_i, \ldots)},
      \qquad 1 \le i \le N - 1,
      \label{eq:qKZi}
    \\
    & \widetilde{K}_1(x_1^{-1}) \ket{\Psi(x_1^{-1}, x_2, \ldots)} = \ket{\Psi(s x_1, x_2, \ldots) },
      \label{eq:qKZ1}
    \\
    & \overline{K}_N(x_N) \ket{\Psi(\ldots, x_{N-1}, x_N)} = \ket{\Psi(\ldots, x_{N-1}, 1/x_N) }.
      \label{eq:qKZN}
\end{align}
Note that the Yang--Baxter, reflection, and unitary conditions ensure the consistency of this definition.  We
will refer to \eqref{eq:qKZi} -- \eqref{eq:qKZN}  as the $q$KZ equations, although in our notation the $q$ has
been replaced by the parameter $s$.  These $q$-difference equations were first introduced in \cite{FrenkelR92}
and appear as $q$-deformation of the KZ equations \cite{KnizhnikZ84}.

Motivated by the connection to the ASEP stationary state, we make the following definition:
\begin{definition}
We call a solution
\begin{equation*}
    \ket{\Psi(\mathbf{x}; s, \xi)} = \sum_{\bm\tau} \psi_{\bm\tau}(\mathbf{x}; s, \xi) \ket{\bm\tau}
\end{equation*}
of equations \eqref{eq:qKZi} -- \eqref{eq:qKZN} a twice deformed inhomogeneous ground state vector, with deformation
parameters $s$ and $\xi$.
\end{definition}
As indicated at the end of the previous section, such a vector with $s = \xi = 1$ is the inhomogeneous ground
state vector of the open boundary ASEP, and can be constructed in matrix product form \cite{CrampeRV14,CrampeMRV15inhomogeneous} or
from specialised non-symmetric Koornwinder polynomials \cite{CantiniGDGW16}.  We will show that more general
solutions exist when $s$ and $\xi$ are related in certain ways.

We use the Noumi representation of the Hecke algebra to write the $q$KZ equations in component form.  To
specify a lattice configuration $\bm\tau$ we use `$\circ$' for an empty site ($\tau_i = 0$) and `$\bullet$'
for a filled site ($\tau_i = 1$).  Then, for example, we write $\psi_{\circ\ldots}$ to indicate the weight for
any configuration with the first site empty ($\tau_1 = 0$)
\begin{lemma}
The $q$KZ equations \eqref{eq:qKZi} -- \eqref{eq:qKZN} for the deformed ground state vector are equivalent to
the following exchange relations on the components:
\begin{align}
    T_0 \psi_{\circ \ldots} & = \xi^{-1} t_0^{-1/2} \psi_{\bullet \ldots},
    \label{eq:qkzExchangeLeft}
    \\
    T_N \psi_{\ldots \bullet} & = t_N^{-1/2} \psi_{\ldots \circ},
    \label{eq:qkzExchangeRight}
\end{align}
and for $1 \le i \le N - 1$,
\begin{align}
    T_i \psi_{\ldots \circ \circ \ldots} & = t^{1/2} \psi_{\ldots \circ \circ \ldots},
    \label{eq:qkzExchangeBulk00}
    \\
    T_i \psi_{\ldots \bullet \bullet \ldots} & = t^{1/2} \psi_{\ldots \bullet \bullet \ldots},
    \label{eq:qkzExchangeBulk11}
    \\
    T_i \psi_{\ldots \bullet \circ \ldots} & = t^{-1/2} \psi_{\ldots \circ \bullet \ldots},
    \label{eq:qkzExchangeBulk01}
\end{align}
where the marked sites are in positions $i$, $i+1$.
\label{lemma:rightExchange}
\end{lemma}
\begin{proof}
This can be checked directly.
\end{proof}
Note that the parameters $s$ and $\xi$ both enter through \eqref{eq:qkzExchangeLeft}, with $s$ contained
within the $T_0$ operator.

\begin{lemma}
For any vector $\ket{\Psi(\mathbf{x}; s, \xi)}$ satisfying the $q$KZ equations \eqref{eq:qKZi} --
\eqref{eq:qKZN}, the empty lattice weight $\psi_{\circ \ldots \circ}$ is an eigenfunction of the $Y_i$
operators \eqref{eq:Yi}, satisfying
\begin{equation}
    Y_i \psi_{\circ \ldots \circ}
        = \xi^{-1} t_0^{-1/2} t_N^{-1/2} t^{-(i-1)} \psi_{\circ \ldots \circ}.
    \label{eq:asepEigenref}
\end{equation}
\label{lemma:asepEigenref}
\end{lemma}
\begin{proof}
This follows by direct computation with the exchange relations in Lemma~\ref{lemma:rightExchange}.
\end{proof}

Lemma~\ref{lemma:asepEigenref} immediately suggests the connection to the non-symmetric Koornwinder polynomials:
\begin{enumerate}
    \item
Taking $\xi = s^m$, $m > 0$, the eigenvalue in \eqref{eq:asepEigenref} is given by \eqref{eq:ynegm},
corresponding to the non-symmetric Koornwinder polynomial labelled by the composition $\left((-m)^N \right)$.

    \item
Taking $\xi = t_0^{-1} t_N^{-1} t^{-(N-1)} s^{-m}$, $m \ge 0$, the eigenvalue instead corresponds to
\eqref{eq:yposm}, for the composition $\left(m^N\right)$.
\end{enumerate}
Moreover, note that case $2$ is obtained from case $1$ by sending
\begin{equation*}
    \xi \to t_0^{-1} t_N^{-1} t^{-(N-1)} \xi^{-1},
\end{equation*}
which is exactly the Gallavotti--Cohen symmetry \eqref{eq:GCsym}.
  In section~\ref{sec:matrixproduct} we will
give a direct matrix product construction of the inhomogeneous ground state vector for case 1, that is $\xi =
s^m$.  To solve case 2, we will use the Gallavotti--Cohen symmetry on solutions of \emph{left} $q$KZ
equations, which we will present next.
We note that an alternative approach, as followed in \cite{Kasatani10,CantiniGDGW16}, would be to take $\psi_{\circ\ldots\circ}$
as the non-symmetric Koornwinder polynomial given in case 1 or case 2, then show that a solution of the
exchange relations \eqref{eq:qkzExchangeLeft} -- \eqref{eq:qkzExchangeBulk01} can be constructed from this
reference state.

\subsection{Left $q$KZ equations}
We define left $q$KZ equations
\begin{align}
    & \bra{\Phi(\ldots, x_i, x_{i+1}, \ldots)} \check{R}_i(x_{i+1}/x_{i}) = \bra{\Phi(\ldots, x_{i+1}, x_i, \ldots)},
      \qquad 1 \le i \le N - 1,
      \label{eq:leftqKZi}
    \\
    & \bra{\Phi(x_1^{-1}, x_2, \ldots)} \widetilde{K}_1(x_1^{-1}) = \bra{\Phi(s x_1, x_2, \ldots) },
      \label{eq:leftqKZ1}
    \\
    & \bra{\Phi(\ldots, x_{N-1}, x_N)} \overline{K}_N(x_N) = \bra{\Phi(\ldots, x_{N-1}, 1/x_N) },
      \label{eq:leftqKZN}
\end{align}
with
\begin{equation}
    \bra{\Phi(\mathbf{x})} = \sum_{\bm\tau} \phi_{\bm\tau}(\mathbf{x}) \bra{\bm\tau}.
    \label{eq:vecPhi}
\end{equation}
These would imply a solution of a left scattering equation (analogous to \eqref{eq:scattering}) with a
scattering matrix, defined by reversing the order of matrices in the definition \eqref{eq:scatteringmatrix}.
The two following lemmas are analogous to lemmas \ref{lemma:rightExchange} and \ref{lemma:asepEigenref}.

\begin{lemma}
The left $q$KZ equations \eqref{eq:leftqKZi} -- \eqref{eq:leftqKZN} for a vector of form \eqref{eq:vecPhi}
are equivalent to the following exchange relations on the components:
\begin{align}
    T_0 \phi_{\circ \ldots} & = \xi t_0^{1/2} \phi_{\bullet \ldots},
    \label{eq:leftqkzExchangeLeft}
    \\
    T_N \phi_{\ldots \bullet} & = t_N^{1/2} \phi_{\ldots \circ},
    \label{eq:leftqkzExchangeRight}
\end{align}
and for $1 \le i \le N - 1$,
\begin{align}
    T_i \phi_{\ldots \circ \circ \ldots} & = t^{1/2} \phi_{\ldots \circ \circ \ldots},
    \label{eq:leftqkzExchangeBulk00}
    \\
    T_i \phi_{\ldots \bullet \bullet \ldots} & = t^{1/2} \phi_{\ldots \bullet \bullet \ldots},
    \label{eq:leftqkzExchangeBulk11}
    \\
    T_i \phi_{\ldots \bullet \circ \ldots} & = t^{1/2} \phi_{\ldots \circ \bullet \ldots},
    \label{eq:leftqkzExchangeBulk01}
\end{align}
where the marked sites are in positions $i$, $i+1$.
\label{lemma:leftqkzExchange}
\end{lemma}

\begin{lemma}
For any vector $\bra{\Phi(\mathbf{x}; s, \xi)}$ satisfying the left $q$KZ equations \eqref{eq:leftqKZi} --
\eqref{eq:leftqKZN}, the empty lattice weight $\phi_{\circ \ldots \circ}$ is an eigenfunction of the $Y_i$
operators \eqref{eq:Yi}, satisfying
\begin{equation}
    Y_i \phi_{\circ \ldots \circ}
        = \xi t_0^{1/2} t_N^{1/2} t^{N-i} \phi_{\circ \ldots \circ}.
    \label{eq:asepleftEigenref}
\end{equation}
\label{lemma:asepleftEigenref}
\end{lemma}

Again, the same two constraints on $\xi$ and $s$ appear, but with the correspondence to the non-symmetric
Koornwinder polynomials reversed: taking $\xi = s^m$, $m \ge 0$, would correspond to the composition
$\left(m^N\right)$; taking $\xi = t_0^{-1} t_N^{-1} t^{-(N-1)} s^{-m}$, $m > 0$, would correspond to the
composition $\left((-m)^N)\right)$.

The Gallavotti--Cohen symmetry allows us to relate solutions of the left and right $q$KZ equations.
\begin{lemma}
For any vector $\bra{\Phi(\mathbf{x}; s, \xi)}$ satisfying the left $q$KZ equations \eqref{eq:leftqKZi} --
\eqref{eq:leftqKZN}, the vector
\begin{equation}
    \ket{\Psi(\mathbf{x}; s, \xi')} = U_{\text{GC}} \ket{\Phi(\mathbf{x}; s, \xi)},
\end{equation}
with
\begin{equation}
    \xi' = t_0^{-1} t_N^{-1} t^{-(N-1)} \xi^{-1},
\end{equation}
is a solution to the right $q$KZ equations \eqref{eq:qKZi} -- \eqref{eq:qKZN}.
\label{lemma:leftrightsol}
\end{lemma}
\begin{proof}
This is checked by transposing the left $q$KZ equations and using the Gallavotti--Cohen symmetry on the
$\check{R}$, $\widetilde{K}$ and $\overline{K}$ matrices.
\end{proof}

\section{Matrix product solution} 
\label{sec:matrixproduct}
 
The matrix product ansatz for the stationary state of the ASEP was introduced in \cite{DerridaEHP93} and has led since then 
to numerous works in statistical physics and mathematical physics.
The connection with integrability was explored in \cite{SasamotoW97,CrampeRV14} and allowed the generalisation 
to the multi-species open ASEP \cite{CrampeMRV15,CrampeEMRV16,CrampeFRV16}. In \cite{LazarescuM11,GorissenLMV12,Lazarescu13jphysA} a perturbative
matrix ansatz was constructed giving successive cumulants of the particle current, and some of
the structures introduced there will be used here.  For the periodic ASEP, the matrix product method was used
to construct the stationary state of the multi-species system \cite{EvansFM09,ProlhacEM09,AritaAMP11,AritaAMP12} 
and the  Macdonald polynomials \cite{CantiniDGW15}, and has revealed a connection to the 3D integrability of the model \cite{KunibaMO15,KunibaMO16}.
In this section we give a matrix product construction
of the twice deformed ground state vectors, and show that this results in a matrix product formula for certain
symmetric Koornwinder polynomials.

\subsection{General construction}

The matrix product ansatz for the twice deformed inhomogeneous ground state vectors is written
\begin{equation}
    \ket{\Psi(\mathbf{x}; s, \xi)} = \bbra{W} \mathbb{S} \mathbb{A}_1(x_1) \ldots \mathbb{A}_N(x_N) \kket{V},
    \label{eq:mpaForm}
\end{equation}
with
\begin{equation*}
    \mathbb{A}(x) = \begin{pmatrix}
        A_0(x)
        \\
        A_1(x)
    \end{pmatrix}.
\end{equation*}
The entries $A_0(x)$, $A_1(x)$ as well as $\mathbb{S}$ are operators in some auxiliary algebraic space,
and the left and right vectors $\bbra{W}$ and $\kket{V}$ contract this space to give a scalar value.
The indices in \eqref{eq:mpaForm} denote the lattice site the vector $\mathbb{A}(x)$ relates to, and in
tensor product notation, we would write
\begin{equation*}
    \ket{\Psi(\mathbf{x}; s, \xi)}
        = \bbra{W} \mathbb{S} \mathbb{A}(x_1) \otimes \ldots \otimes \mathbb{A}(x_N) \kket{V}.
\end{equation*}
However, as a matter of convention, we will reserve the symbol `$\otimes$' for objects belonging to the
auxiliary algebraic space (see \eqref{eq:Sm} for example), and 
use the index notation to denote the tensor product in
the space of lattice configurations.  The aim is to provide a notational distinction between these two spaces.
Writing out
\eqref{eq:mpaForm} gives the $2^N$ component vector
\begin{equation*}
    \ket{\Psi(\mathbf{x}; s, \xi)}
    = \begin{pmatrix}
        \bbra{W} \mathbb{S} A_0(x_1) \ldots A_0(x_{N-1}) A_0(x_N) \kket{V}
        \\
        \bbra{W} \mathbb{S} A_0(x_1) \ldots A_0(x_{N-1}) A_1(x_N) \kket{V}
        \\
        \bbra{W} \mathbb{S} A_0(x_1) \ldots A_1(x_{N-1}) A_0(x_N) \kket{V}
        \\
        \vdots
        \\
        \bbra{W} \mathbb{S} A_1(x_1) \ldots A_1(x_{N-1}) A_1(x_N) \kket{V}
    \end{pmatrix},
\end{equation*}
with entries
\begin{equation*}
\psi_{\bm\tau}(\mathbf{x}; s, \xi)=   \bbra{W} \mathbb{S} A_{\tau_1}(x_1) \ldots A_{\tau_{N-1}}(x_{N-1}) A_{\tau_N}(x_N) \kket{V}.
\end{equation*}

\begin{lemma}
Sufficient conditions for a vector of form \eqref{eq:mpaForm} to satisfy the $q$KZ equations \eqref{eq:qKZi} --
\eqref{eq:qKZN} are the following:
\begin{align}
    \check{R}\left(\frac{x_{i+1}}{x_i}\right) \mathbb{A}_1(x_i) \mathbb{A}_2(x_{i+1})
    & =
    \mathbb{A}_1(x_{i+1}) \mathbb{A}_2(x_i),
    \label{eq:ZF}
    \\
    \widetilde{K}\left(x_1^{-1}\right) \bbra{W} \mathbb{S} \mathbb{A}\left(x_1^{-1}\right)
    & =
    \bbra{W} \mathbb{S} \mathbb{A}\left(s x_1\right),
    \label{eq:deformedGZ1}
    \\
    \overline{K}(x_N) \mathbb{A}(x_N) \kket{V}
    & =
    \mathbb{A}\left(x_N^{-1}\right) \kket{V}.
    \label{eq:deformedGZN}
\end{align}
\end{lemma}
Equation \eqref{eq:ZF} is the Zamolodchikov--Faddeev (ZF) algebra \cite{ZamolodchikovZ79,Faddeev80}.  Equations
\eqref{eq:deformedGZ1}, \eqref{eq:deformedGZN} are a deformation of the Ghoshal--Zamolodchikov (GZ) relations
\cite{GhoshalZ94}.  The undeformed GZ relations are obtained by setting $\mathbb{S}$ to the identity and $s = 1$.
The matrix product ansatz for the open boundary ASEP can be expressed as a solution of the undeformed
relations, and solutions for related models have also been found and studied \cite{CrampeRV14,CrampeRRV16}.

\subsection{Construction of solutions}
We now give an explicit construction of the $q$KZ solution when $\xi = s^m$, $m \ge 1$.  We first define certain algebraic objects through the relations they satisfy.
\begin{definition}
We define algebraic objects satisfying the following relations: operators $a$, $a^\dagger$ and $S$:
\begin{equation}
\begin{aligned}
    & a a^\dagger - t a^\dagger a = 1 - t,
    \\
    & a S = \sqrt{s} S a,
    \\
    & S a^\dagger = \sqrt{s} a^\dagger S.
\end{aligned}
\label{eq:oscAlgBulk}
\end{equation}
And paired boundary vectors $\bbra{w}$ and $\kket{v}$:
\begin{equation}
\begin{aligned}
    \bbra{w} \left( t_0^{1/2}a-t_0^{-1/2}a^\dagger \right)
        & = \bbra{w} \left( u_0^{1/2}-u_0^{-1/2} \right),
  \\
    \left( t_N^{1/2}a^\dagger-t_N^{-1/2}a \right) \kket{v}
        & =  \left( u_N^{1/2}-u_N^{-1/2} \right) \kket{v},
\label{eq:oscAlgBoundary}
\end{aligned}
\end{equation}
and $\bbra{\widetilde{w}}$ and $\kket{\widetilde{v}}$:
\begin{equation}
\begin{aligned}
    \bbra{\widetilde w} \left( t_0^{1/2}a-t_0^{-1/2}a^\dagger \right)
        & = \bbra{\widetilde w} \left( t_0^{1/2}-t_0^{-1/2} \right),
        \\
    \left( t_N^{1/2}a^\dagger-t_N^{-1/2} a \right) \kket{\widetilde v}
        & =  \left( t_N^{1/2}-t_N^{-1/2} \right)\kket{\widetilde v}.
\end{aligned}
\label{eq:oscAlgBoundaryTwid}
\end{equation}
\label{def:oscAlg}
\end{definition}
Elements of this algebra have appeared in many places in the context of the ASEP. 
The first algebraic relation of \eqref{eq:oscAlgBulk} and the relations \eqref{eq:oscAlgBoundary} were first stated in \cite{Sandow94} 
to study the stationary state of the open ASEP. This work shed new light on the DEHP algebra introduced in \cite{DerridaEHP93} by 
showing that it can be recast in a form of a $q$-deformed oscillator algebra by an appropriate shift and normalisation of the generators.
The representation of the algebraic elements involved in the first relation of \eqref{eq:oscAlgBulk} and in the
relations \eqref{eq:oscAlgBoundary} were found in \cite{Sandow94}, and permitted explicit computations. In
particular the author of that work
pointed out the relevance of the parametrisation used here. More precisely the parameters $\kappa_+(\alpha,\gamma)$ and $\kappa_+(\beta,\delta)$ 
defined in \cite{Sandow94} by $\kappa_+(x,y)=\frac{1}{2x}\left(y-x+p-q+\sqrt{(y-x+p-q)^2+4xy}\right)$ play a central role in the representation of the 
algebra, and are relevant in describing the phase transitions of the system. The precise relations with the parameters used here are
$\kappa_+(\alpha,\gamma)=u_0^{1/2}t_0^{-1/2}$ and $\kappa_+(\beta,\delta)=u_N^{1/2}t_N^{-1/2}$.

The other relations
\eqref{eq:oscAlgBulk}, \eqref{eq:oscAlgBoundary} and \eqref{eq:oscAlgBoundaryTwid} appear previously in
\cite{GorissenLMV12,Lazarescu13jphysA,Lazarescu13,LazarescuP14} to compute the fluctuations of the current. 
In Appendix~\ref{app:representation} we recall
an infinite dimensional representation of this algebra: $\kket{v},\kket{\widetilde v},\dots$ are vectors of a Fock space endowed 
with the usual scalar product. In this paper the scalar product of two vectors $\kket{x}$ and $\kket{y}$ of this Fock space is denoted
by $\bbra{x} \cdot \kket{y}$.
The operators $a$ and $a^\dagger$ are linear operators on this Fock space.
Let us stress here that the creation operator $a^\dagger$ is not the Hermitian conjugate of the annihilation operator $a$ (it is 
a standard notation which appears often in the literature, see for instance \cite{Sandow94}).  

Building on this algebra, we define
\begin{align}
    \mathbb{S}^{(m)} & = S^{2m-1} \otimes S^{2m-2} \otimes \ldots \otimes S^3 \otimes S^2 \otimes S,
    \label{eq:Sm}
    \\
    \mathbb{A}^{(m)}(x)
        & = \underbrace{L(x) \dot\otimes \ldots \dot\otimes L(x)}_{m-1 \text{ times}} \dot\otimes b(x),
    \label{eq:Am}
\end{align}
with
\begin{equation}
    L(x) = \begin{pmatrix}
        1 & a
        \\
        x a^\dagger & x
    \end{pmatrix}\dot\otimes \begin{pmatrix}
        1/x & a/x
        \\
         a^\dagger & 1
    \end{pmatrix},
    \qquad
    b(x) = \begin{pmatrix}
        1/x+a
        \\
        x+a^\dagger
    \end{pmatrix}.
    \label{eq:Lbdef}
\end{equation}
The symbol $\dot\otimes$ indicates the normal dot product in the physical space, taking the tensor product of
the entries -- elements of the auxiliary algebraic space.
For example, expanding the definition of $L(x)$ gives
\begin{equation*}
    L(x) = \begin{pmatrix}
        x^{-1} \id \otimes \id + a \otimes a^\dagger
        &
        x^{-1} \id \otimes a + a \otimes \id
        \\
        a^\dagger \otimes \id + x \id \otimes a^\dagger
        &
        a^\dagger \otimes a + x \id \otimes \id
    \end{pmatrix}.
\end{equation*}
We also define boundary vectors
\begin{align}
    \bbra{W^{(m)}}
    & =
    \underbrace{
        \bbra{w} \otimes \bbra{\widetilde w} \otimes \ldots \otimes \bbra{w} \otimes \bbra{\widetilde w}
    }_{m-1 \text{ times}} \otimes \bbra{w}
    \label{eq:Wm}
    \\
    \kket{V^{(m)}}
    &=
    \underbrace{
        |v\rrangle \otimes \kket{\widetilde v} \otimes \ldots \otimes \kket{v} \otimes \kket{\widetilde v}
    }_{m-1 \text{ times}} \otimes \kket{v}.
    \label{eq:Vm}
\end{align}

\begin{proposition}
For integer $m > 0$ and $\xi = s^m$, 
\begin{equation}
    \ket{\Psi^{(m)}(\mathbf{x};s)}
    =
    \frac{
        1
    }
    {
        \Omega^{(m)}
    }
    \bbra{W^{(m)}}
        \mathbb{S}^{(m)} \mathbb{A}^{(m)}_1(x_1) \ldots \mathbb{A}^{(m)}_N(x_N) 
    \kket{V^{(m)}},
    \label{eq:Psim}
\end{equation}
with normalisation factor
\begin{equation}
    \Omega^{(m)} = \bbra{W^{(m)}} \mathbb{S}^{(m)} \kket{V^{(m)}},
    \label{eq:Omegam}
\end{equation}
is a solution of the $q$KZ equations \eqref{eq:qKZi} -- \eqref{eq:qKZN}.
\end{proposition}
Note that the dependence on $\xi$ has disappeared in the vector $\ket{\Psi^{(m)}(\mathbf{x};s)}$ because of the constraint $\xi = s^m$. 
\begin{proof}
The normalisation factor $\Omega^{(m)}$ can be chosen freely, but we must show that the choice
\eqref{eq:Omegam} is non-zero.  To do so, we compute $\Omega^{(m)}$ using an infinite dimensional
representation of the algebra defined in \eqref{eq:oscAlgBulk} -- \eqref{eq:oscAlgBoundaryTwid}.  We give the
details in Appendix~\ref{app:representation}.  Then to prove that $\ket{\Psi^{(m)}(\mathbf{x};s)}$ is a $q$KZ
solution, it is sufficient to show that \eqref{eq:ZF}, \eqref{eq:deformedGZ1}, \eqref{eq:deformedGZN} are
satisfied.

By a direct computation, using the algebraic relations \eqref{eq:oscAlgBulk}, it can be checked that the
vector $b(x)$ and the matrix $L(x)$ satisfy the relations
\begin{align*}
 \check{R}(x_{i+1}/x_{i})b_1(x_{i})b_2(x_{i+1}) & = b_1(x_{i+1}) b_2(x_{i}),
 \\
 \check{R}(x_{i+1}/x_{i}) L_1(x_{i}) L_2(x_{i+1}) & = L_1(x_{i+1}) L_2(x_{i})\check{R}(x_{i+1}/x_{i}).
\end{align*}
These elementary exchange relations can be used successively several times to give \eqref{eq:ZF}.
On the right boundary, using relations \eqref{eq:oscAlgBoundary},
\eqref{eq:oscAlgBoundaryTwid} gives
\begin{align*}
 \overline{K}(x_N) b(x_N)|v\rrangle  & = b(1/x_N)|v\rrangle,
  \\
 \overline{K}(x_N) L(x_N)|v\rrangle \otimes | \widetilde v \rrangle
    & = L(1/x_N)\overline{K}(x_N)|v\rrangle \otimes | \widetilde v \rrangle.
\end{align*}
Using these properties several times, it is straightforward to prove \eqref{eq:deformedGZN}.
Finally, on the left boundary, the vector $b(x)$ satisfies
\begin{equation*}
 \llangle w| S \left.\widetilde{K}(x_1^{-1})\right|_{\xi=s} b(x_1^{-1}) = \llangle w| S b(s x_1),
\end{equation*}
and the matrix $L(x)$ satisfies
\begin{align*}
 \llangle w| \otimes \llangle \widetilde w| & S^{2a+1} \otimes S^{2a} \left.\widetilde{K}(x_1^{-1})\right|_{\xi=s^{a+1}} L(x_1^{-1})
 \\
 & = 
 \llangle w| \otimes \llangle \widetilde w| S^{2a+1} \otimes S^{2a} L(s x_1) \left.\widetilde{K}(x_1^{-1})\right|_{\xi=s^a}.
\end{align*}
In words, the last equation means that the parameter $\xi$ is multiplied by a factor $s$ when the matrix $L$
passes through the matrix $\widetilde{K}$.  Thus by imposing the constraint $\xi = s^m$ and applying these
relations successively, relation \eqref{eq:deformedGZ1} follows.
\end{proof}
We still need to show that the construction gives a non-zero vector.  Before doing so, we introduce some
notation, then look at some examples.

\begin{definition}
For a lattice configuration $\bm\tau = (\tau_1, \ldots, \tau_N)$, define the composition
$\lambda^{(m)}(\bm\tau)$, with
\begin{equation*}
    \lambda^{(m)}(\bm\tau)_i = \begin{cases}
        -m, & \tau_i = 0,
        \\
        m,  & \tau_i = 1.
    \end{cases}
\end{equation*}
The corresponding partition is $\lambda^{(m)+}(\bm\tau) = \left(m^N\right)$.
\end{definition}
\begin{definition}
We introduce the notation
\begin{align*}
    ^{k}\langle B \rangle & = \bbra{w} S^k B \kket{v},
    \\
    ^{j,k}\langle B \dot\otimes C \rangle
    & =
    \left(\bbra{w}\otimes\bbra{\widetilde{w}}\right)
         \left(S^j \otimes S^k \right) \left(B \dot\otimes C \right)
    \left(\kket{v} \otimes \kket{\widetilde{v}}\right)
    \\
    & =
    \Big(\bbra{w} S^j B \kket{v}\Big)
    .
    \left(\bbra{\widetilde{w}} S^k C \kket{\widetilde{v}}\right)
\end{align*}
Here $B$, $C$ may be scalars, vectors, or matrices in physical space, with entries belonging to the auxiliary
algebraic space.
\end{definition}

\begin{example}
With $m=1$, $\mathbb{A}^{(1)}(x) = b(x)$.  For $N=1$,
\begin{equation*}
\begin{aligned}
    \Omega^{(1)} \ket{\Psi^{(1)}(x_1; s)}
    & =
    {^{1}\langle} b(x_1)_1 \rangle
    \\
    & =
    \begin{pmatrix}
        \bbra{w} S \left(\frac{1}{x_1} + a \right)\kket{v}
        \\
        \bbra{w} S \left(x_1 + a^\dagger \right)\kket{v}
    \end{pmatrix},
\end{aligned}
\end{equation*}
and for $N=2$,
\begin{equation*}
\begin{aligned}
    \Omega^{(1)} \ket{\Psi^{(1)}(x_1, x_2; s)}
    & =
    {^1\langle} b(x_1)_1  b(x_2)_2 \rangle
    \\
    & =
    \begin{pmatrix}
        \bbra{w} S \left(\frac{1}{x_1} + a\right) \left(\frac{1}{x_2} + a \right) \kket{v}
        \\
        \bbra{w} S \left(\frac{1}{x_1} + a\right) \left(x_2 + a^\dagger \right) \kket{v}
        \\
        \bbra{w} S \left(x_1 + a^\dagger\right) \left(\frac{1}{x_2} + a \right) \kket{v}
        \\
        \bbra{w} S \left(x_1 + a^\dagger\right) \left(x_2 + a^\dagger \right) \kket{v}
    \end{pmatrix}.
\end{aligned}
\end{equation*}
In general,
\begin{equation*}
    \Omega^{(1)} \ket{\Psi^{(1)}(\mathbf{x}; s)}
    =
    {^1\langle} b(x_1)_1  \ldots b(x_N)_N \rangle.
\end{equation*}
Note that the normalisation  $\Omega^{(1)} = \bbra{w} S \kket{v}$ ensures that each component $\psi^{(1)}_{\bm\tau}$
has leading term $\mathbf{x}^{\lambda^{(1)}(\bm\tau)}$ with coefficient $1$, and all other terms correspond
to compositions $\mu$ with
\begin{equation*}
    \mu^+ < \lambda^{(1)+}(\bm\tau) = \left(1^N\right).
\end{equation*}
\label{ex:Psim1}
\end{example}

\begin{example}
With $m=2$,
\begin{equation*}
    \mathbb{A}^{(2)}(x) = L(x) \dot\otimes b(x).
\end{equation*}
Then for $N=1$,
\begin{align*}
    \Omega^{(2)}\ket{\Psi^{(2)}(x_1; s)}
    & =
    \left(\bbra{w}\otimes\bbra{\widetilde{w}}\otimes{\bbra{w}}\right)
    \left(S^3 \otimes S^2 \otimes S\right)
    \left(L(x_1) \dot\otimes b(x_1)\right)_1
    \left(\kket{v}\otimes\kket{\widetilde{v}}\otimes\kket{v}\right)
    \\
    & =
    {^{3,2}\langle} L(x_1)_1 \rangle . \left({^1\langle} b(x_1)_1 \rangle\right)
    \\
    & =
    {^{3,2}\langle} L(x_1)_1 \rangle . \Omega^{(1)} \ket{\Psi^{(1)}(x_1; s)},
\end{align*}
with
\begin{equation*}
    {^{3,2}\langle} L(x_1)_1 \rangle
    =
    \bbra{w}\otimes\bbra{\widetilde{w}}S^3\otimes S^2
    \begin{pmatrix}
        x_1^{-1} \id \otimes \id + a \otimes a^\dagger
        &
        x_1^{-1} \id \otimes a + a \otimes \id
        \\
        a^\dagger \otimes \id + x_1 \id \otimes a^\dagger
        &
        a^\dagger \otimes a + x_1 \id \otimes \id
    \end{pmatrix}_1
    \kket{v} \otimes \ket{\widetilde{v}}.
\end{equation*}
For $N=2$,
\begin{align*}
    \Omega^{(2)}\ket{\Psi^{(2)}(x_1; s)}
    =
    \left(\bbra{w}\otimes\bbra{\widetilde{w}}\otimes{\bbra{w}}\right)
    \left(S^3 \otimes S^2 \otimes S\right)
    & \left(L(x_1) \dot\otimes b(x_1)\right)_1
    \\
    & . \left(L(x_2) \dot\otimes b(x_2)\right)_2
    \left(\kket{v}\otimes\kket{\widetilde{v}}\otimes\kket{v}\right).
\end{align*}
The matrix $L(x_2)_2$ can be brought past $b(x_1)_1$ as they are in different physical spaces, and
their entries are in different auxiliary algebraic spaces.  Thus we obtain
\begin{equation*}
    \Omega^{(2)}\ket{\Psi^{(2)}(x_1, x_2; s)}
    =
    {^{3,2}\langle} L(x_1)_1 L(x_2)_2\rangle . \left(\Omega^{(1)} \ket{\Psi^{(1)}(x_1, x_2; s)} \right),
\end{equation*}
with
\begin{equation*}
\begin{aligned}
    {^{3,2}\langle} L(x_1)_1 L(x_2)_2 \rangle
    =
    \bbra{w}\otimes\bbra{\widetilde{w}} S^3 \otimes S^2
    & \begin{pmatrix}
        x_1^{-1} \id \otimes \id + a \otimes a^\dagger
        &
        x_1^{-1} \id \otimes a + a \otimes \id
        \\
        a^\dagger \otimes \id + x_1 \id \otimes a^\dagger
        &
        a^\dagger \otimes a + x_1 \id \otimes \id
    \end{pmatrix}_1
    \\
    & . \begin{pmatrix}
        x_2^{-1} \id \otimes \id + a \otimes a^\dagger
        &
        x_2^{-1} \id \otimes a + a \otimes \id
        \\
        a^\dagger \otimes \id + x_2 \id \otimes a^\dagger
        &
        a^\dagger \otimes a + x_2 \id \otimes \id
    \end{pmatrix}_2
    \kket{v}\otimes\kket{\widetilde{v}}.
\end{aligned}
\end{equation*}
The normalisation factor is
\begin{equation*}
    \Omega^{(2)}
    = \bbra{w}\otimes\bbra{\widetilde{w}}\otimes\bbra{w}
        S^3 \otimes S^2 \otimes S \kket{v}\otimes\kket{\widetilde{v}} \otimes \kket{v},
\end{equation*}
and it can be checked directly for $N=1,2$ that
each component $\psi^{(2)}_{\bm\tau}$
has leading term $\mathbf{x}^{\lambda^{(2)}(\bm\tau)}$ with coefficient $1$, and all other terms correspond
to compositions $\mu$ with
\begin{equation*}
    \mu^+ < \lambda^{(2)+}(\bm\tau) = \left(2^N\right).
\end{equation*}
\label{ex:Psim2}
\end{example}
We now give the general form.
\begin{theorem}
\label{theorem:Psisol}
The $q$KZ equations have a solution when $\xi = s^m$, written recursively on $m$:
For $m > 1$
\begin{equation}
    \ket{\Psi^{(m)}(\mathbf{x};s)}
    =
    \frac{1}{\bbra{w}S^{2m-1}\kket{v} \bbra{\tilde{w}}S^{2m-2}\kket{\tilde{v}}}
    \left(
        {^{2m-1,2m-2}\langle} L(x_1)_1 \ldots L(x_N)_N \rangle
    \right)
    \ket{\Psi^{(m-1)}(\mathbf{x};s)},
    \label{eq:Psimrec}
\end{equation}
with
\begin{equation}
    \ket{\Psi^{(1)}(\mathbf{x};s)}
    =
    \frac{1}{\bbra{w} S \kket{v}}\left({^1\langle}b(x_1)_1 \ldots b(x_N)\rangle\right).
    \label{eq:Psi1rec}
\end{equation}
The components of the solution, $\psi^{(m)}_{\bm\tau}(\mathbf{x};s)$, have leading term
$\mathbf{x}^{\lambda^{(m)}(\bm\tau)}$, and all other terms correspond to compositions $\mu$ with
\begin{equation*}
    \mu^+ < \lambda^{(m)+}(\bm\tau) = \left(m^N\right).
\end{equation*}
\end{theorem}
\begin{proof}
The recursive form \eqref{eq:Psimrec}, \eqref{eq:Psi1rec} is obtained by a reordering of the matrix product
form \eqref{eq:Psim}, as in Example~\ref{ex:Psim2}.

The second part of the claim, on the degree and normalisation of components of the solution, can be proven
inductively.  We assume the property holds at $m-1$ and use \eqref{eq:Psimrec} to obtain the solution at $m$.
That is, we multiply by the `increment' matrix
\begin{equation*}
    \frac{1}{\bbra{w}S^{2m-1}\kket{v} \bbra{\tilde{w}}S^{2m-2}\kket{\tilde{v}}}
    \left(
        {^{2m-1,2m-2}\langle} L(x_1)_1 \ldots L(x_N)_N \rangle
    \right).
\end{equation*}
The following points can be deduced by writing \eqref{eq:Psimrec} and the increment matrix in component form:
\begin{itemize}
    \item
A term $\mathbf{x}^\mu$ with $\mu^+ = \left(m^N\right)$ can only be produced from the leading order terms
of the $m-1$ solution, which correspond to the partition $\left((m-1)^N\right)$, and thus we can ignore
sub-leading terms.

    \item
The diagonal entries of the increment matrix produce the term $\mathbf{x}^{\lambda^{(m)}(\bm\tau)}$ with
coefficient $1$ (plus lower order terms) in $\psi^{(m)}_{\bm\tau}$, from the corresponding
component of the $m-1$ solution.

    \item
The off-diagonal entries of the increment matrix, acting on the leading order term of a component of the $m-1$ solution,
either reduces the degree or leaves it unchanged.
\end{itemize}
These points are sufficient to show that the degree and normalisation properties hold at $m$, assuming they hold
at $m-1$.  As the $m=1$ case was checked in Example~\ref{ex:Psim1}, the properties hold for all $m$.
\end{proof}

The construction of the solution of the left $q$KZ equations at $s = \xi^m$ is similar to the above, and we
defer the details to Appendix~\ref{app:leftgroundstate}.  We state here the main result.

\begin{theorem}
For integer $m \ge 0$ and $\xi = s^m$, solutions of the left $q$KZ equations can be constructed in matrix
product form, and can be defined recursively.  For $m > 0$
\begin{equation}
    \bra{\Phi^{(m)}(\mathbf{x};s)}
    =
    \frac{1}{\bbra{w} S^{2m-1} \kket{v} \bbra{\tilde{w}} S^{2m} \kket{\tilde{v}}}
    \bra{\Phi^{(m-1)}(\mathbf{x};s)}
    \left(
        {^{2m-1,2m}\langle} L\left(\frac{1}{x_1}\right)_1 \ldots L\left(\frac{1}{x_N}\right)_N \rangle
    \right)
    \label{eq:Phimrec}
\end{equation}
with
\begin{equation}
    \bra{\Phi^{(0)}(\mathbf{x};s)} = \bra{1} = (1, 1)^{\otimes N}
\end{equation}
The solution is non-zero: the component of the solution, $\phi^{(m)}_{\bm\tau}(\mathbf{x};s)$, contains the term
$\mathbf{x}^{-\lambda^{(m)}(\bm\tau)}$ with coefficient $1$, and all terms correspond to compositions $\mu$ with
\begin{equation*}
    \mu^+ \le \lambda^{(m)+}(\bm\tau) = \left(m^N\right).
\end{equation*}

\label{theorem:Phisol}
\end{theorem}

\begin{corollary}
For $m > 0$ and $\xi = t_0^{-1} t_N^{-1} t^{-(N-1)} s^{-m}$, the right $q$KZ equations
\eqref{eq:qKZi} -- \eqref{eq:qKZN} have solution
\begin{equation}
    \ket{\Psi(\mathbf{x}; s, \xi = t_0^{-1} t_N^{-1} t^{-(N-1)} s^{-m})}
    =
    U_{\text{GC}} \ket{\Phi^{(m)}(\mathbf{x}; s)}.
\end{equation}
\label{corr:PsiPhisol}
\end{corollary}
\begin{proof}
This follows from Lemma~\ref{lemma:leftrightsol}.
\end{proof}

The $m=0$ case is a bit special:
$\ket{\Psi(\mathbf{x}; s, \xi = t_0^{-1} t_N^{-1} t^{-(N-1)})} = U_{\text{GC}} \ket{1}$. The solution does not depend on 
$\mathbf{x}$ and $s$. Imposing in addition that $\xi=1$, i.e. $t_0t_Nt^{N-1}=1$, gives
a very simple ASEP stationary state. Indeed the system is at thermal equilibrium in this case: written in the usual ASEP parameters
the constraint is $\frac{\alpha\beta}{\gamma\delta}\left(\frac{p}{q}\right)^{N-1}=1$.

\subsection{Symmetric Koornwinder polynomials}
We can now make the connection between solutions of the $q$KZ equations, and the symmetric and non-symmetric
Koornwinder polynomials.

\begin{lemma}
The component $\psi^{(m)}_{\circ \ldots \circ}$ of the vector $\ket{\Psi^{(m)}(\mathbf{x},s)}$ is the
non-symmetric Koornwinder polynomial $E_{((-m)^N)}$.  All other components can be constructed through the
relations
\begin{equation}
\begin{aligned}
    \psi^{(m)}_{\circ \ldots \circ \bullet} & = t_N^{-1/2} T_N^{-1} \psi^{(m)}_{\circ \ldots \circ \circ},
    \\
    \psi^{(m)}_{\ldots \bullet \circ \ldots}
    &
    = t^{-1/2} T_i^{-1} \psi^{(m)}_{\ldots \circ \bullet \ldots},
    \qquad 1 \le i \le N-1.
\end{aligned}
\label{eq:constructpsis}
\end{equation}
The set of all components $\{\psi^{(m)}_{\bm\tau}\}$ forms a basis for $\mathcal{R}^{(m^N)}$, the space
spanned by non-symmetric Koornwinder polynomials $\{E_\mu | \mu \in \mathbb{Z}^N, \mu^+ = (m^N)\}$.
\label{lemma:psibasis}
\end{lemma}
\begin{proof}
By Theorem~\ref{theorem:Psisol}, and Lemma~\ref{lemma:asepEigenref} with $\xi = s^m$, $\psi^{(m)}_{\circ
\ldots \circ}$ is an eigenfunction of the $Y_i$ operators, and is a Laurent polynomial with the required
degree and normalisation.  Thus by uniqueness, we can identify $\psi^{(m)}_{\circ \ldots \circ} =
E_{((-m)^N)}.$  The relations \eqref{eq:constructpsis} come from the exchange relations
\eqref{eq:qkzExchangeBulk01} and \eqref{eq:qkzExchangeRight}.

The preceding parts of this lemma give the preconditions for
Proposition 1 and Corollary 1 of \cite{CantiniGDGW16}, from which it follows that
$\{\psi^{(m)}_{\bm\tau}\}$ forms a basis for $\mathcal{R}^{(m^N)}$.
\end{proof}

\begin{lemma}
The component $\phi^{(m)}_{\circ \ldots \circ}$ of the vector $\bra{\Phi^{(m)}(\mathbf{x},s)}$ is the
non-symmetric Koornwinder polynomial $E_{(m^N)}$.  All other components can be constructed through the
relations
\begin{align*}
    \phi^{(m)}_{\circ \ldots \circ \bullet} & = t_N^{1/2} T_N^{-1} \phi^{(m)}_{\circ \ldots \circ \circ},
    \\
    \phi^{(m)}_{\ldots \bullet \circ \ldots}
    &
    = t^{1/2} T_i^{-1} \phi^{(m)}_{\ldots \circ \bullet \ldots},
    \qquad 1 \le i \le N-1.
\end{align*}
The set of all components $\{\phi^{(m)}_{\bm\tau}\}$ forms a basis for $\mathcal{R}^{(m^N)}$.
\end{lemma}
\begin{proof}
This follows in the same way, with reference to Theorem~\ref{theorem:Phisol}, and Lemmas
\ref{lemma:asepleftEigenref} and \ref{lemma:leftqkzExchange}.
\end{proof}

\begin{lemma}
Given a solution $\ket{\Psi(\mathbf{x}; s, \xi)}$ of the $q$KZ equations \eqref{eq:qKZi} --
\eqref{eq:qKZN}, the sum of components
\begin{equation}
    \mathcal{Z}(\mathbf{x}; s, \xi) = \langle 1 | \Psi(\mathbf{x}; s, \xi) \rangle
\end{equation}
is $W_0$ invariant.
\label{lemma:ZW0}
\end{lemma}
\begin{proof}
We first note that $\bra{1}$ is a left eigenvector of $\check{R}_i$, $\overline{K}_N$ with eigenvalue 1 (see
\eqref{eq:Rcheck}, \eqref{eq:Kbar}).  Then applying $\bra{1}$ to the bulk and right boundary $q$KZ equations
\eqref{eq:qKZi}, \eqref{eq:qKZN} we see that $\mathcal{Z}(\mathbf{x}; s, \xi)$ is invariant under $s_i$, $1
\le i \le N$, and hence is $W_0$ invariant.
\end{proof}

\begin{theorem} \label{thm:symmetricKoornAsNormalisation}
The sum of components of $\ket{\Psi^{(m)}(\mathbf{x}; s)}$ is the
symmetric Koornwinder polynomial $P_{\left(m^N\right)}$.  That is
\begin{equation}
    P_{\left(m^N\right)}(\mathbf{x}) = \mathcal{Z}^{(m)}(\mathbf{x}; s),
\end{equation}
where
\begin{equation}
    \mathcal{Z}^{(m)}(\mathbf{x}; s)= \langle 1 | \Psi^{(m)}(\mathbf{x}; s) \rangle,
\end{equation}
and $\ket{\Psi^{(m)}(\mathbf{x}; s)}$ is the $q$KZ solution with $\xi = s^m$, constructed as in
Theorem~\ref{theorem:Psisol}.
\end{theorem}
\begin{proof}
By Lemmas \ref{lemma:psibasis} and \ref{lemma:ZW0}, $\mathcal{Z}^{(m)}(\mathbf{x},; s)$ is $W_0$ invariant,
and belongs to the space $\mathcal{R}^{(m^N)}$, and from Theorem~\ref{theorem:Psisol}, we see that it contains
$\mathbf{x}^{(m^N)}$ with coefficient 1.  The result then follows from the characterisation of symmetric
Koornwinder polynomials in \cite{Sahi99}, quoted in Theorem~\ref{theorem:SahiSymmetric}.
\end{proof}

We note that Theorem~\ref{thm:symmetricKoornAsNormalisation} implies a matrix product construction for the
symmetric Koornwinder polynomial $P_{(m^N)}$.  Direct computations from this form would be difficult, but the
structure leads to certain conjectures that we discuss in Section~\ref{sec:currentKoornwinder}.  We also note
that an integral form for the polynomial $P_{(m^N)}$ is already known \cite{Mimachi01}.  The solution of the
left $q$KZ equation is also related to the same symmetric Koornwinder polynomial.

\begin{theorem} \label{thm:symmetricKoornAsNormalisation2}
The sum of components of $U_{\text{GC}} \ket{\Phi^{(m)}(\mathbf{x}; s)}$ is proportional
to the symmetric Koornwinder polynomial $P_{\left(m^N\right)}$.  That is
\begin{equation}
    P_{\left(m^N\right)}(\mathbf{x}) \propto \mathcal{Z}^{(m)}(\mathbf{x}; s, \xi'),
\end{equation}
where
\begin{equation}
    \mathcal{Z}^{(m)}(\mathbf{x}; s, \xi')
        = \langle 1 | U_{\text{GC}} \ket{\Phi^{(m)}(\mathbf{x}; s)} \rangle,
\end{equation}
with $\xi' = t_0^{-1} t_N^{-1} t^{-(N-1)} s^{-m}$.
\end{theorem}
\begin{proof}
Note that $U_{\text{GC}} \ket{\Phi^{(m)}(\mathbf{x}; s)} \rangle$ is the solution of the right $q$KZ equations
at $\xi'$, and the proof follows as in Theorem~\ref{thm:symmetricKoornAsNormalisation}.  However, because of
the structure of the components $\phi^{(m)}_{\bm\tau}$ (see Theorem~\ref{theorem:Phisol} and
Appendix~\ref{app:leftgroundstate}), and the multiplication by matrix $U_{\text{GC}}$, the coefficient of
$\mathbf{x}^{(m^N)}$ in $\mathcal{Z}^{(m)}(\mathbf{x}; s, \xi')$ is not 1.  Thus the identification with
$P_{(m^N)}$ can only be made up to normalisation.
\end{proof}

\section{Current fluctuations and Koornwinder polynomials}
\label{sec:currentKoornwinder}
 
The aim of this section is to make contact between the machinery developed previously, and the generating
function of the cumulants of the current. The idea is quite simple and arises from the following observation:
the constraint $\xi=s^m$ that was imposed in order to solve the $q$KZ equations, can be satisfied by setting
$s=\xi^{1/m}$, leaving $\xi$ free instead of $s$, which then implies $s \to 1$ as $m \to \infty$.  It appears
then natural to think that the scattering relation \eqref{eq:scattering} may degenerate, in this $s \to 1$
limit, to an eigenvector equation.  Then as $m \to \infty$, the vector
$\ket{\Psi^{(m)}(\mathbf{x};s=\xi^{1/m})}$ should thus converge in some sense to an eigenvector of the
scattering matrix.  To move towards this direction, we make the following conjectures.
 
\begin{conjecture}
 It is conjectured that 
 \begin{eqnarray}
  & & \lim\limits_{m\rightarrow \infty}\frac{\ket{\Psi^{(m)}(\mathbf{x};s=\xi^{1/m})}}{\mathcal{Z}^{(m)}(\mathbf{x};s=\xi^{1/m})}
  = \ket{\Psi_0(\mathbf{x};\xi)}, \\
  & & \lim\limits_{m\rightarrow \infty} \frac{\ln(\xi)}{m}\ln \left(\mathcal{Z}^{(m)}(\mathbf{x};s=\xi^{1/m})\right) 
  = F_0(\mathbf{x};\xi),
 \end{eqnarray}
 with $\ket{\Psi_0}$ and $F_0$  regular functions of $\mathbf{x}$.
\end{conjecture}
These conjectures are supported by strong numerical evidences (up to 3 sites) and by the fact that the matrix
product construction of $\ket{\Psi^{(m)}(\mathbf{x};s)}$ is similar to the one presented in
\cite{GorissenLMV12,Lazarescu13jphysA}. In those works, the authors developed a method called the
``perturbative matrix ansatz'', which allowed them to approximate the ground state of $M(\xi)$, at any order
in the current counting parameter $\xi$. Let us also mention that these kind of results have already been
observed in the context of $q$KZ equations of different models, and are known as the ``quasi-classical'' limit
\cite{Reshetikhin1992,ReshetikhinV94}.
 
Note that the second part of the conjecture can be immediately rewritten in terms of symmetric Koornwinder
polynomials as
\begin{equation}
  \lim\limits_{m\rightarrow \infty} \frac{\ln(\xi)}{m}\ln \left(P_{(m^N)}(\mathbf{x};s=\xi^{1/m})\right) 
  = F_0(\mathbf{x};\xi).
\end{equation}

In the following these conjectures will be considered as facts and properties will be deduced from them. But
one has to keep in mind that the validity of the deduced results relies obviously on the validity of these
conjectures.
 
 \begin{proposition}
  The function $F_0(\mathbf{x};\xi)$ is $W_0$ invariant and its derivative, with respect to any of the $x_i$,
  is invariant under the Gallavotti-Cohen symmetry
  $\xi \to \xi' = t_0^{-1}t_N^{-1}t^{-(N-1)} \xi^{-1}$.
 \end{proposition}
 \begin{proof}
  The $W_0$ invariance directly follows from the $W_0$ invariance of the symmetric Koornwinder polynomials. 
  The Gallavotti-Cohen symmetry follows from theorems \ref{thm:symmetricKoornAsNormalisation} and \ref{thm:symmetricKoornAsNormalisation2}
  which give that 
\begin{align*}
    P_{(m^N)}(\mathbf{x};s=\zeta^{1/m}) & = \mathcal{Z}^{(m)}(\mathbf{x}; s = \zeta^{1/m}, \xi = \zeta)
    \\
    & \propto \mathcal{Z}^{(m)}(\mathbf{x}; s = \zeta^{1/m}, \xi = t_0^{-1} t_N^{-1} t^{-(N-1)} \zeta^{-1})
\end{align*}
  with a proportionality coefficient independent 
  of $\mathbf{x}$. Taking the large $m$ limit it translates into the fact that $F_0(\mathbf{x};\xi')=F_0(\mathbf{x};\xi)+c$ with $c$ a constant 
  term independent of $\mathbf{x}$, which concludes the proof.
 \end{proof}

 In the following we will specify when needed the dependence on $s$ and $\xi$ of the scattering matrix $\mathcal{S}_i(\mathbf{x};s,\xi)$ defined
 in \eqref{eq:scatteringmatrix}.
 
 \begin{proposition}
 The vector $\ket{\Psi_0(\mathbf{x};\xi)}$ is an eigenvector of the scattering matrices evaluated at $s=1$,
 with
 \begin{equation} \label{eq:eigenvector_scatteringmatrices}
  \mathcal{S}_i(\mathbf{x};1,\xi)\ket{\Psi_0(\mathbf{x};\xi)}=
  \exp\left( x_i \frac{\partial F_0}{\partial x_i}(\mathbf{x};\xi)\right) \ket{\Psi_0(\mathbf{x};\xi)}.
 \end{equation}
 \end{proposition}
 \begin{proof}
  Our starting point is the scattering relation \eqref{eq:scattering} applied with $s=\xi^{1/m}$.  We divide by 
  $\mathcal{Z}^{(m)}(\mathbf{x};s=\xi^{1/m})$ to obtain
  \begin{eqnarray*}
  & &  \mathcal{S}_i(\mathbf{x};s=\xi^{1/m},\xi) 
  \frac{\ket{\Psi^{(m)}(\ldots, x_i, \ldots;s=\xi^{1/m})}}{\mathcal{Z}^{(m)}(\ldots, x_i, \ldots;s=\xi^{1/m})} \\
  & & \hspace{1cm} = \frac{\mathcal{Z}^{(m)}(\ldots, \xi^{1/m}x_i, \ldots;s=\xi^{1/m})}{\mathcal{Z}^{(m)}(\ldots, x_i, \ldots;s=\xi^{1/m})}
  \frac{\ket{\Psi^{(m)}(\ldots, \xi^{1/m}x_i, \ldots;s=\xi^{1/m})}}{\mathcal{Z}^{(m)}(\ldots, \xi^{1/m}x_i, \ldots;s=\xi^{1/m})}.
\end{eqnarray*}
We then have the limits 
 \begin{eqnarray*}
  & & \lim\limits_{m\rightarrow \infty}\mathcal{S}_i(\mathbf{x};s=\xi^{1/m},\xi)
  =\mathcal{S}_i(\mathbf{x};1,\xi), \\
  & & \lim\limits_{m\rightarrow \infty}\frac{\ket{\Psi^{(m)}(\ldots, x_i, \ldots;s=\xi^{1/m})}}{\mathcal{Z}^{(m)}(\ldots, x_i, \ldots;s=\xi^{1/m})}
  = \ket{\Psi_0(\ldots, x_i, \ldots;\xi)}, \\
  & & \lim\limits_{m\rightarrow \infty}\frac{\ket{\Psi^{(m)}(\ldots, \xi^{1/m}x_i, \ldots;s=\xi^{1/m})}}{\mathcal{Z}^{(m)}(\ldots, \xi^{1/m}x_i, \ldots;s=\xi^{1/m})}
  = \ket{\Psi_0(\ldots, x_i, \ldots;\xi)}, \\
  & & \lim\limits_{m\rightarrow \infty} \frac{\mathcal{Z}^{(m)}(\ldots, \xi^{1/m}x_i, \ldots;s=\xi^{1/m})}{\mathcal{Z}^{(m)}(\ldots, x_i, \ldots;s=\xi^{1/m})}
  = \exp\left( x_i \frac{\partial F_0}{\partial x_i}(\mathbf{x};\xi)\right),
 \end{eqnarray*}
 which yield the desired result.
 \end{proof}

 \begin{proposition}
 The vector $\ket{\Psi_0(\mathbf{1};\xi)}$ is an eigenvector of the deformed Markov matrix, with
 \begin{equation}
  M(\xi)\ket{\Psi_0(\mathbf{1};\xi)}=\frac{p-q}{2}\frac{\partial^2 F_0}{\partial
  x_i^2}(\mathbf{1};\xi)\ket{\Psi_0(\mathbf{1};\xi)}.
 \end{equation}
This implies immediately the following expression for the generating function of the cumulants of the
 current:
 \begin{equation}
  E(\mu)=\Lambda_0(e^{\mu})=\frac{p-q}{2}\frac{\partial^2 F_0}{\partial x_i^2}(\mathbf{1};e^{\mu}).
 \end{equation}
 \end{proposition}
\begin{proof}
This is proven by taking the derivative of \eqref{eq:eigenvector_scatteringmatrices} with respect to $x_i$ and
then setting $x_1=\dots=x_N=1$. One has to make basic use of the properties given in
\eqref{eq:properties_scatteringmatrices}, and notice the fact that $\frac{\partial F_0}{\partial
x_i}(\mathbf{1};\xi)=0$ because $F_0$ is $W_0$ invariant.
\end{proof}

Note that despite revealing a beautiful connection between the symmetric functions and the fluctuations of the
current, in practice, the last expression does not help to compute the cumulants of the current because a
closed expression for $F_0$ is missing.
 
However a step can be made toward an exact expression of $F_0$, using the characterization of the symmetric Koornwinder polynomials
 as eigenfunctions of the finite difference operator $D$ defined in \eqref{eq:operatorD}. 
 The eigenvalue $d_{(m^N)}$ of this operator associated to the symmetric Koornwinder polynomial $P_{(m^N)}$ is given for $s=\xi^{1/m}$ by
 \begin{equation}
 d_0(\xi)=\frac{1-t^N}{1-t}(\xi-1)(t_0t_Nt^{N-1}-1/\xi).
 \end{equation}
It is straightforward to check that $d_0(\xi)$ is invariant under the Gallavotti-Cohen symmetry, that is $d_0(\xi)=d_0(\xi')$.

In the following we will explicitly write the dependence on $s$ of the functions $g_i(\mathbf{x};s)$ defined in \eqref{eq:function_g_i}.
\begin{proposition}
We have the following characterization of the function $F_0$:
\begin{eqnarray*}
 & &\sum_{i=1}^{N} g_i(\mathbf{x};1)\left[\exp\left( x_i \frac{\partial F_0}{\partial x_i}(\mathbf{x};\xi)\right)-1\right]
 +\sum_{i=1}^{N} g_i(\mathbf{x}^{-1};1)\left[\exp\left(- x_i \frac{\partial F_0}{\partial x_i}(\mathbf{x};\xi)\right)-1\right] \\
 & & =d_0(\xi)
\end{eqnarray*}
\end{proposition}
\begin{proof}
 This follows directly from the relation \eqref{eq:eigenvalue_eq_symmetric_Koornwinder}
 applied for the symmetric Koornwinder polynomial $P_{(m^N)}$ and $s=\xi^{1/m}$. Dividing the 
 latter relation by $P_{(m^N)}(\mathbf{x};s=\xi^{1/m})$ and taking the large $m$ limit yield the desired result.
\end{proof}
It would be interesting to understand if this characterization of the function $F_0$ can be related to a Baxter $t-Q$ relation \cite{Baxter82}.
It has been shown in \cite{CrampeMRV15inhomogeneous} that the normalisation of the stationary state (corresponding to the case $s=\xi=1$)
satisfies a $t-Q$ difference equation.
We can mention also in this context the work \cite{LazarescuP14} where the authors constructed a Baxter Q operator for the open ASEP
(with current-counting deformation) and derived the corresponding $t-Q$ relations.

\section*{Acknowledgments}
We would like to warmly thank Nicolas Cramp{\'e} and Eric Ragoucy
for stimulating discussions and suggestions.

\appendix

\section{Calculation of the normalisation}
\label{app:representation}

We wish to show that the normalisation factor $\Omega^{(m)}$, defined in \eqref{eq:Omegam}, is non-zero.  To do so we
show how it is computed in the infinite dimensional representation of the algebra given in Definition~\ref{def:oscAlg}.
Our presentation of the representation follows the review in \cite{Lazarescu13}, and we refer to that work for further
references.

The representation of the algebra is defined on the Fock space $\text{Span}\{ \kket{k} \}_{k=0}^\infty$.  The
bulk matrices are given by
\begin{align*}
    & a = \sum_{k=1}^\infty (1 - t^k) \kket{k - 1} \bbra{k},
    \qquad
    a^\dagger = \sum_{k=0}^\infty \kket{k + 1} \bbra{k},
    \\
    & S  = \sum_{k=0}^\infty s^{k/2} \kket{k} \bbra{k}.
\end{align*}
The boundary vectors are written
\begin{equation*}
    \bbra{w} = \sum_{k=0}^\infty w_k \bbra{k},
    \qquad
    \kket{v} = \sum_{k=0}^\infty v_k \kket{k},
\end{equation*}
and
\begin{equation*}
    \bbra{\widetilde{w}} = \sum_{k=0}^\infty \widetilde{w}_k \bbra{k},
    \qquad
    \kket{\widetilde{v}} = \sum_{k=0}^\infty \widetilde{v}_k \kket{k}.
\end{equation*}
As a consequence of the boundary relations, the coefficients appearing in $\bbra{w}$, $\kket{v}$ satisfy the
recursion relations
\begin{align*}
    & w_{k+1} + t_0^{1/2}(u_0^{1/2} - u_0^{-1/2}) w_k - t_0(1-t^k)w_{k-1} = 0,
    \\
    & (t)_{k+1} v_{k+1} + t_N^{1/2}(u_N^{1/2} - u_N^{-1/2}) (t)_k v_k - t_N (1-t^k) (t)_{k-1} v_{k-1} = 0,
\end{align*}
with $w_{-1} = v_{-1} = 0$.  We have used the $t$-Pochhammer symbol
\begin{equation*}
    (x)_n = \prod_{k=0}^{n-1} (1 - t^k x).
\end{equation*}
The $t$-Pochhammer symbol can be defined for $n = \infty$ if $t < 1$, then
\begin{equation*}
    (x)_\infty = \prod_{k=0}^\infty (1 - t^k x),
\end{equation*}
and we use the notation
\begin{equation*}
    (x, y, z, \ldots)_\infty = (x)_\infty (y)_\infty (z)_\infty \ldots
\end{equation*}
The $t$-Hermite polynomials are given by
\begin{equation*}
    H_n(x, y) = \sum_{k=0}^n \frac{(t)_n}{(t)_k (t)_{n-k}} x^k y^{n-k},
\end{equation*}
and satisfy the recursion relation
\begin{equation*}
    H_{n+1}(x, y) - (x + y) H_n(x, y) + x y (1 - t^n) H_{n-1}(x, y) = 0.
\end{equation*}
Thus we find
\begin{equation}
\begin{aligned}
    w_n & = H_n(t_0^{1/2} u_0^{-1/2}, -t_0^{1/2} u_0^{1/2}),
    \\
    v_n & = \frac{H_n(t_N^{1/2} u_N^{-1/2}, -t_N^{1/2} u_N^{1/2})}{(t)_n}.
\end{aligned}
\label{eq:wnvn}
\end{equation}
The coefficients $\widetilde{w}_n$, $\widetilde{v}_n$ are obtained by setting $u_i = t_i$ in the above.

To compute the normalisation factors in the $q$KZ solution, we will make use of the $t$-Mehler formula
\begin{equation*}
    \sum_{n=0}^\infty H_n(x, y) H_n(w, z) \frac{\lambda^n}{(q)_n}
    =
    \frac{(x y w z \lambda^2)_\infty}{(x w \lambda, x z \lambda, y w \lambda, y z \lambda)_\infty}.
\end{equation*}
For the normalisation, we need to compute
\begin{equation*}
    \bbra{w} S^a \kket{v} = \sum_{n=0}^\infty w_n \left(s^{n /2}\right)^a v_n.
\end{equation*}
Using the coefficients \eqref{eq:wnvn} and the $t$-Mehler formula gives
\begin{equation}
\begin{aligned}
   &  \bbra{w} S^a \kket{v} = 
   \\
   & \frac{\left(t_0^{1/2} t_N^{1/2} s^a\right)_\infty}
    {\left(
        t_0^{1/2} u_0^{-1/2} t_N^{1/2} u_N^{-1/2} s^{a/2},
        -t_0^{1/2} u_0^{-1/2} t_N^{1/2} u_N^{1/2} s^{a/2},
        -t_0^{1/2} u_0^{1/2} t_N^{1/2} u_N^{-1/2} s^{a/2},
        t_0^{1/2} u_0^{1/2} t_N^{1/2} u_N^{1/2} s^{a/2}
    \right)_\infty}.
\end{aligned}
\label{eq:infNormFactor}
\end{equation}
We also need to compute $\bbra{\widetilde{w}} S^a \kket{\widetilde{v}}$, but this is obtained from
\eqref{eq:infNormFactor} by setting $u_i = t_i$, $i = 0, N$. 

\section{Construction of a left ground state vector}
\label{app:leftgroundstate}

In this appendix, we construct row vector solutions of the \emph{left} $q$KZ equations \eqref{eq:leftqKZi} --
\eqref{eq:leftqKZN} in the matrix product form
\begin{equation}
    \bra{\Phi(\mathbf{x}; s, \xi)} = \bbra{W} \mathbb{S} \mathbb{A}_1(x_1) \ldots \mathbb{A}_N(x_N) \kket{V},
    \label{eq:leftmpaForm}
\end{equation}
with
\begin{equation}
    \mathbb{A}(x) = \begin{pmatrix}
        A_0(x)
        &
        A_1(x)
    \end{pmatrix}.
\end{equation}

\begin{lemma}
Sufficient conditions for a vector of form \eqref{eq:leftmpaForm} to satisfy the left $q$KZ equations
\eqref{eq:leftqKZi} -- \eqref{eq:leftqKZN} are the following:
\begin{align}
    \mathbb{A}_1(x_i) \mathbb{A}_2(x_{i+1}) \check{R}\left(\frac{x_{i+1}}{x_i}\right) 
    & =
    \mathbb{A}_1(x_{i+1}) \mathbb{A}_2(x_i),
    \label{eq:leftZF}
    \\
    \bbra{W} \mathbb{S} \mathbb{A}\left(x_1^{-1}\right)\widetilde{K}\left(x_1^{-1}\right) 
    & =
    \bbra{W} \mathbb{S} \mathbb{A}\left(s x_1\right),
    \label{eq:leftdeformedGZ1}
    \\
    \mathbb{A}(x_N) \kket{V} \overline{K}(x_N)
    & =
    \mathbb{A}\left(x_N^{-1}\right) \kket{V}.
    \label{eq:leftdeformedGZN}
\end{align}
\end{lemma}

We will construct a solution at $\xi = s^m$, with $m \ge 0$ an integer.  We define the following objects:
\begin{equation}
\begin{aligned}
    b_{\text{left}}(x) & = (1, 1),
    \\
    \mathbb{A}^{(m)}_{\text{left}}(x)
        & = b_{\text{left}}(x) \dot\otimes \underbrace{L(1/x) \dot\otimes \ldots \dot\otimes L(1/x)}_{m \text{ times}},
    \\
    \mathbb{S}^{(m)}_{\text{left}} & = 1 \otimes S \otimes S^2 \otimes \ldots \otimes S^{2m-1} \otimes S^{2m},
    \\
    \bbra{W^{(m)}_{\text{left}}}
    & = \bbra{\widetilde w} \otimes \underbrace{\bbra{w} \otimes \bbra{\widetilde w}
        \otimes \ldots \otimes \bbra{w} \otimes \bbra{\widetilde w}}_{m \text{ times}},
    \\
    \kket{V^{(m)}_{\text{left}}}
    & = \kket{\widetilde v } \otimes \underbrace{\kket{v} \otimes \kket{\widetilde v}
        \otimes \ldots \otimes \kket{v} \otimes \kket{\widetilde v}}_{m \text{ times}}.
\end{aligned}
\end{equation}
The algebraic objects are as given in Definition~\ref{def:oscAlg}, and $L(x)$ is defined in \eqref{eq:Lbdef}.

\begin{proposition}
For integer $m \ge 0$ and $\xi = s^m$, 
\begin{equation}
    \bra{\Phi^{(m)}(\mathbf{x}; s)}
    =
    \frac{1}{\Omega^{(m)}_{\text{left}}}
    \bbra{W^{(m)}_{\text{left}}}
        \mathbb{S}^{(m)}_{\text{left}} \mathbb{A}^{(m)}_{\text{left},1}(x_1)
        \ldots
        \mathbb{A}^{(m)}_{\text{left},N}(x_N) 
    \kket{V^{(m)}_{\text{left}}},
    \label{eq:Phim}
\end{equation}
with normalisation factor
\begin{equation}
    \Omega^{(m)}_{\text{left}}
        = \bbra{W^{(m)}_{\text{left}}} \mathbb{S}^{(m)}_{\text{left}} \kket{V^{(m)}_{\text{left}}},
\end{equation}
is a solution of left the $q$KZ equations \eqref{eq:leftqKZi} -- \eqref{eq:leftqKZN}.
\end{proposition}
\begin{proof}
We give the elementary exchange relations, which imply \eqref{eq:leftZF} -- \eqref{eq:leftdeformedGZN},
and thus a solution of the left $q$KZ equations.

In the bulk,
\begin{align*}
    b_{\text{left},1}(x_i) b_{\text{left},2}(x_{i+1})\check{R}(x_{i+1}/x_{i})
    & =
    b_{\text{left},1}(x_{i+1}) b_{\text{left},2}(x_i),
    \\
    L_1(1/x_{i}) L_2(1/x_{i+1})\check{R}(x_{i+1}/x_{i})  & = \check{R}(x_{i+1}/x_{i}) L_1(1/x_{i+1}) L_2(1/x_{i}),
\end{align*}
from which \eqref{eq:leftZF} follows.  On the right boundary
\begin{align*}
    b_{\text{left}}(x_N) \overline{K}(x_N)\kket{v}  & =  b_{\text{left}}(1/x_N) \kket{v},
    \\
 L(1/x_N)\overline{K}(x_N)|v\rrangle \otimes | \widetilde v \rrangle & = \overline{K}(x_N)L(x_N)|v\rrangle \otimes | \widetilde v \rrangle,
\end{align*}
from which \eqref{eq:leftdeformedGZN} follows.  On the left boundary
\begin{equation*}
 \bbra{w}  b_{\text{left}}(1/x_1) \left.\widetilde{K}(x_1^{-1})\right|_{\xi=1}
 =
 \bbra{w} b_{\text{left}}(s x_1),
\end{equation*}
and 
\begin{align*}
 & \llangle w| \otimes \llangle \widetilde w| S^{2a-1} \otimes S^{2a} L(x_1) \left.\widetilde{K}(x_1^{-1})\right|_{\xi=s^{a}}
 \\
 & = 
 \llangle w| \otimes \llangle \widetilde w| S^{2a-1} \otimes S^{2a} \left.\widetilde{K}(x_1^{-1})\right|_{\xi=s^{a-1}} L(1/(s x_1)).
\end{align*}
With the constraint $\xi = s^m$, these imply \eqref{eq:leftdeformedGZ1}.
\end{proof}

With this result, the proof of Theorem~\ref{theorem:Phisol} is then very similar to that for
Theorem~\ref{theorem:Psisol}.  We will, however, comment briefly on the terms appearing in each component
$\phi^{(m)}_{\bm\tau}(\mathbf{x})$, and the normalisation.  To do so, we look at the normalised `increment
matrix' taking the $m-1$ solution to the $m$ solution:
\begin{equation*}
    \frac{1}{\bbra{w} S^{2m-1} \kket{v} \bbra{\widetilde{w}} S^{2m} \kket{\widetilde{v}}}
    \left(
        {^{2m-1,2m}\langle} L\left(\frac{1}{x_1}\right)_1 \ldots L\left(\frac{1}{x_N}\right)_N \rangle
    \right).
\end{equation*}
The diagonal entries of this matrix contain the term $\mathbf{x}^{-\lambda^{(1)}(\bm\tau)}$ with coefficient
$1$, which produce the term $\mathbf{x}^{-\lambda^{(m)}(\bm\tau)}$ in $\phi^{(m)}_{\bm\tau}$.  In the top row of
the increment matrix, the entry in column $\bm\tau'$ has leading term
\begin{equation}
    \frac{\bbra{\widetilde w}S^{2m}{a^{\sum_i \tau'_i}} \kket{\widetilde v}}
         {\bbra{\widetilde{w}} S^{2m} \kket{\widetilde{v}}} x_1 \ldots x_N,
    \label{eq:xaaacoeff}
\end{equation}
and as a consequence, each component $\phi^{(m)}_{\bm\tau}$ contains the term $\mathbf{x}^{(m^N)}$ with the same
coefficient as in \eqref{eq:xaaacoeff}.


\begin{thebibliography}{10}

\bibitem{Spitzer70}
F.~Spitzer, ``{Interaction of {M}arkov processes},''
  \href{http://dx.doi.org/10.1016/0001-8708(70)90034-4}{{\em Advances in
  Mathematics} {\bfseries 5} no.~2, (1970) 246--290}.

\bibitem{Liggett85}
T.~M. Liggett, {\em {Interacting Particle Systems}}, vol.~276.
\newblock Springer-Verlag, New York, 1985.

\bibitem{Derrida98}
B.~Derrida, ``{An exactly soluble non-equilibrium system: the asymmetric simple
  exclusion process},''
  \href{http://dx.doi.org/10.1016/S0370-1573(98)00006-4}{{\em Physics Reports}
  {\bfseries 301} no.~1, (1998) 65--83}.

\bibitem{ChouMZ11}
T.~Chou, K.~Mallick, and R.~Zia, ``{Non-equilibrium statistical mechanics: from
  a paradigmatic model to biological transport},''
  \href{http://dx.doi.org/10.1088/0034-4885/74/11/116601}{{\em Reports on
  progress in physics} {\bfseries 74} no.~11, (2011) 116601},
  \href{http://arxiv.org/abs/1110.1783}{{\ttfamily arXiv:1110.1783
  [cond-mat.stat-mech]}}.

\bibitem{KatzLS84}
S.~Katz, J.~L. Lebowitz, and H.~Spohn, ``{Nonequilibrium steady states of
  stochastic lattice gas models of fast ionic conductors},''
  \href{http://dx.doi.org/10.1007/BF01018556}{{\em J Stat Phys} {\bfseries 34}
  (1984) 497}.

\bibitem{KrapivskyRB10}
P.~L. Krapivsky, S.~Redner, and E.~Ben-Naim, {\em {A kinetic view of
  statistical physics}}.
\newblock Cambridge University Press, 2010.

\bibitem{SchmittmannZ95}
B.~Schmittmann and R.~K. Zia, ``{Statistical mechanics of driven diffusive
  systems},'' \href{http://dx.doi.org/10.1016/S1062-7901(06)80014-5}{{\em Phase
  transitions and critical phenomena} {\bfseries 17} (1995) 3--214}.

\bibitem{Touchette09}
H.~Touchette, ``{The large deviation approach to statistical mechanics},''
  \href{http://dx.doi.org/10.1016/j.physrep.2009.05.002}{{\em Physics Reports}
  {\bfseries 478} no.~1--3, (2009) 1--69},
  \href{http://arxiv.org/abs/0804.0327}{{\ttfamily arXiv:0804.0327
  [cond-mat.stat-mech]}}.

\bibitem{DerridaEM95}
B.~Derrida, M.~Evans, and K.~Mallick, ``{Exact diffusion constant of a
  one-dimensional asymmetric exclusion model with open boundaries},''
  \href{http://dx.doi.org/10.1007/BF02181206}{{\em J Stat Phys} {\bfseries 79}
  no.~5-6, (1995) 833}.

\bibitem{deGierE11}
J.~de~Gier and F.~H.~L. Essler, ``{Large Deviation Function for the Current in
  the Open Asymmetric Simple Exclusion Process},''
  \href{http://dx.doi.org/10.1103/PhysRevLett.107.010602}{{\em Phys. Rev.
  Lett.} {\bfseries 107} (2011) 010602},
  \href{http://arxiv.org/abs/1101.3235}{{\ttfamily arXiv:1101.3235
  [cond-mat.stat-mech]}}.

\bibitem{LazarescuM11}
A.~Lazarescu and K.~Mallick, ``{An exact formula for the statistics of the
  current in the TASEP with open boundaries},''
  \href{http://dx.doi.org/10.1088/1751-8113/44/31/315001}{{\em J. Phys. A:
  Math. Theor.} {\bfseries 44} no.~31, (2011) 315001},
  \href{http://arxiv.org/abs/1104.5089}{{\ttfamily arXiv:1104.5089
  [cond-mat.stat-mech]}}.

\bibitem{GorissenLMV12}
M.~Gorissen, A.~Lazarescu, K.~Mallick, and C.~Vanderzande, ``{Exact Current
  Statistics of the Asymmetric Simple Exclusion Process with Open
  Boundaries},'' \href{http://dx.doi.org/10.1103/PhysRevLett.109.170601}{{\em
  Phys. Rev. Lett.} {\bfseries 109} (2012) 170601},
  \href{http://arxiv.org/abs/1207.6879}{{\ttfamily arXiv:1207.6879
  [cond-mat.stat-mech]}}.

\bibitem{Lazarescu13jphysA}
A.~Lazarescu, ``{Matrix ansatz for the fluctuations of the current in the ASEP
  with open boundaries},''
  \href{http://dx.doi.org/10.1088/1751-8113/46/14/145003}{{\em J. Phys. A:
  Math. Theor.} {\bfseries 46} no.~14, (2013) 145003},
  \href{http://arxiv.org/abs/1212.3366}{{\ttfamily arXiv:1212.3366
  [cond-mat.stat-mech]}}.

\bibitem{LazarescuP14}
A.~Lazarescu and V.~Pasquier, ``{Bethe Ansatz and Q -operator for the open
  ASEP},'' \href{http://dx.doi.org/10.1088/1751-8113/47/29/295202}{{\em J.
  Phys. A: Math. Theor.} {\bfseries 47} (2014) 295202},
  \href{http://arxiv.org/abs/1403.6963}{{\ttfamily arXiv:1403.6963 [math-ph]}}.

\bibitem{DerridaEHP93}
B.~Derrida, M.~R. Evans, V.~Hakim, and V.~Pasquier, ``{Exact solution of a 1{D}
  asymmetric exclusion model using a matrix formulation},''
  \href{http://dx.doi.org/10.1088/0305-4470/26/7/011}{{\em Journal of Physics
  A: Mathematical and General} {\bfseries 26} no.~7, (1993) 1493}.

\bibitem{Baxter82}
R.~J. Baxter, {\em {Exactly solved models in statistical mechanics}}.
\newblock Academic Press, London, 1982.

\bibitem{Sklyanin88}
E.~K. Sklyanin, ``{Boundary conditions for integrable quantum systems},''
  \href{http://dx.doi.org/10.1088/0305-4470/21/10/015}{{\em J. Phys. A: Math.
  Gen.} {\bfseries 21} (1988) 2375}.

\bibitem{SasamotoW97}
T.~Sasamoto and M.~Wadati, ``{Stationary state of integrable systems in matrix
  product form},'' \href{http://dx.doi.org/10.1143/JPSJ.66.2618}{{\em J. Phys.
  Soc. Jpn.} {\bfseries 66} no.~9, (1997) 2618--2627}.

\bibitem{CrampeRV14}
N.~Crampe, E.~Ragoucy, and M.~Vanicat, ``{Integrable approach to simple
  exclusion processes with boundaries. Review and progress},''
  \href{http://dx.doi.org/10.1088/1742-5468/2014/11/P11032}{{\em J. Stat.
  Mech.} (2014) P11032}, \href{http://arxiv.org/abs/1408.5357}{{\ttfamily
  arXiv:1408.5357 [math-ph]}}.

\bibitem{CantiniDGW15}
L.~Cantini, J.~de~Gier, and M.~Wheeler, ``{Matrix product formula for
  {M}acdonald polynomials},''
  \href{http://dx.doi.org/10.1088/1751-8113/48/38/384001}{{\em J. Phys. A:
  Math. Theor.} {\bfseries 48} no.~38, (2015) 384001},
  \href{http://arxiv.org/abs/1505.00287}{{\ttfamily arXiv:1505.00287
  [math-ph]}}.

\bibitem{Cantini15}
L.~Cantini, ``{Asymmetric Simple Exclusion Process with open boundaries and
  {K}oornwinder polynomials},''
  \href{http://arxiv.org/abs/1506.00284}{{\ttfamily arXiv:1506.00284
  [math-ph]}}.

\bibitem{CantiniGDGW16}
L.~Cantini, A.~Garbali, J.~de~Gier, and M.~Wheeler, ``{Koornwinder polynomials
  and the stationary multi-species asymmetric exclusion process with open
  boundaries},'' \href{http://dx.doi.org/10.1088/1751-8113/49/44/444002}{{\em
  J. Phys. A: Math. Theor.} {\bfseries 49} no.~44, (2016) 444002},
  \href{http://arxiv.org/abs/1607.00039}{{\ttfamily arXiv:1607.00039
  [math-ph]}}.

\bibitem{SchuetzD93}
G.~Sch{\"u}tz and E.~Domany, ``Phase transitions in an exactly soluble
  one-dimensional exclusion process,''
  \href{http://dx.doi.org/10.1007/BF01048050}{{\em J. Stat. Phys.} {\bfseries
  72} no.~1-2, (1993) 277--296},
  \href{http://arxiv.org/abs/cond-mat/9303038}{{\ttfamily
  arXiv:cond-mat/9303038}}.

\bibitem{OrrS13}
D.~{Orr} and M.~{Shimozono}, ``Specializations of nonsymmetric
  {M}acdonald-{K}oornwinder polynomials,''
  \href{http://arxiv.org/abs/1310.0279}{{\ttfamily arXiv:1310.0279 [math.QA]}}.

\bibitem{GallavottiC95}
G.~Gallavotti and E.~G.~D. Cohen, ``{Dynamical ensembles in stationary
  states},'' \href{http://dx.doi.org/10.1007/BF02179860}{{\em J Stat Phys}
  {\bfseries 80} no.~5, (1995) 931--970}.

\bibitem{LebowitzS99}
J.~L. Lebowitz and H.~Spohn, ``{A {G}allavotti--{C}ohen-Type Symmetry in the
  Large Deviation Functional for Stochastic Dynamics},''
  \href{http://dx.doi.org/10.1023/A:1004589714161}{{\em Journal of Statistical
  Physics} {\bfseries 95} (1999) 333--365},
  \href{http://arxiv.org/abs/cond-mat/9811220}{{\ttfamily
  arXiv:cond-mat/9811220 [cond-mat.stat-mech]}}.

\bibitem{Koornwinder92}
T.~H. Koornwinder, ``{{A}skey--{W}ilson polynomials for root systems of type
  {$BC_n$}},'' {\em Contemp. Math} {\bfseries 138} (1992) 189--204.

\bibitem{vanDiejen96}
J.~van Diejen, ``{Self-dual {K}oornwinder-{M}acdonald polynomials},''
  \href{http://dx.doi.org/10.1007/s002220050102}{{\em J. Invent math}
  {\bfseries 126} (1996) 319--339}.

\bibitem{Sahi99}
S.~Sahi, ``{Nonsymmetric {K}oornwinder Polynomials and Duality},''
  \href{http://dx.doi.org/10.2307/121102}{{\em Annals of Mathematics}
  {\bfseries 150} no.~1, (1999) pp. 267--282},
  \href{http://arxiv.org/abs/q-alg/9710032}{{\ttfamily arXiv:q-alg/9710032}}.

\bibitem{Stokman00}
J.~V. Stokman, ``{{K}oornwinder polynomials and affine {H}ecke algebras},''
  \href{http://dx.doi.org/10.1155/S1073792800000520}{{\em Int Math Res Notices}
  {\bfseries 2000} (2000) 1005--1042},
  \href{http://arxiv.org/abs/math/0002090}{{\ttfamily arXiv:math/0002090
  [math.QA]}}.

\bibitem{Lusztig89}
G.~Lusztig, ``{Affine {H}ecke Algebras and Their Graded Version},''
  \href{http://dx.doi.org/10.2307/1990945}{{\em J. Amer. Math. Soc.} {\bfseries
  2} (1989) 599--635}.

\bibitem{Noumi95}
M.~Noumi, ``{{M}acdonald-{K}oornwinder polynomials and affine Hecke rings},''
  {\em Surikaisekikenkyusho kokyuroku} {\bfseries 919} (1995) 44--55. (In
  Japanese).

\bibitem{ZinnJustin07}
P.~Zinn-Justin, ``{Loop model with mixed boundary conditions, q{KZ} equation
  and alternating sign matrices},''
  \href{http://dx.doi.org/10.1088/1742-5468/2007/01/P01007}{{\em J. Stat.
  Mech.} (2007) P01007}, \href{http://arxiv.org/abs/math-ph/0610067}{{\ttfamily
  arXiv:math-ph/0610067}}.

\bibitem{FrenkelR92}
I.~B. Frenkel and N.~Y. Reshetikhin, ``{Quantum affine algebras and holonomic
  difference equations},'' \href{http://dx.doi.org/10.1007/BF02099206}{{\em
  Commun.Math. Phys.} {\bfseries 146} no.~1, (1992) 1--60}.

\bibitem{KnizhnikZ84}
V.~Knizhnik and A.~Zamolodchikov, ``{Current algebra and Wess-Zumino model in
  two dimensions},'' \href{http://dx.doi.org/10.1016/0550-3213(84)90374-2}{{\em
  Nucl. Phys. B} {\bfseries 247} no.~1, (1984) 83--103}.

\bibitem{CrampeMRV15inhomogeneous}
N.~Crampe, K.~Mallick, E.~Ragoucy, and M.~Vanicat, ``{Inhomogeneous
  discrete-time exclusion processes},''
  \href{http://dx.doi.org/10.1088/1751-8113/48/48/484002}{{\em J. Phys. A:
  Math. Theor.} {\bfseries 48} no.~48, (2015) 484002},
  \href{http://arxiv.org/abs/1506.04874}{{\ttfamily arXiv:1506.04874
  [cond-mat.stat-mech]}}.

\bibitem{Kasatani10}
M.~Kasatani, ``{Boundary quantum {K}nizhnik-{Z}amolodchikov equation},'' in
  {\em {New trends in quantum integrable systems}}, B.~Feigin, M.~Jimbo, and
  M.~Okado, eds., pp.~157--171.
\newblock World Scientific, 2010.

\bibitem{CrampeMRV15}
N.~Crampe, K.~Mallick, E.~Ragoucy, and M.~Vanicat, ``{Open two-species
  exclusion processes with integrable boundaries},''
  \href{http://dx.doi.org/10.1088/1751-8113/48/17/175002}{{\em J. Phys. A:
  Math. Theor.} {\bfseries 48} no.~17, (2015) 175002},
  \href{http://arxiv.org/abs/1412.5939}{{\ttfamily arXiv:1412.5939
  [cond-mat.stat-mech]}}.

\bibitem{CrampeEMRV16}
N.~Crampe, M.~Evans, K.~Mallick, E.~Ragoucy, and M.~Vanicat, ``{Matrix product
  solution to a 2-species {TASEP} with open integrable boundaries},''
  \href{http://dx.doi.org/10.1088/1751-8113/49/47/475001}{{\em J. Phys. A:
  Math. Theor.} {\bfseries 49} (2016) 475001},
  \href{http://arxiv.org/abs/1606.08148}{{\ttfamily arXiv:1606.08148
  [cond-mat.stat-mech]}}.

\bibitem{CrampeFRV16}
N.~Crampe, C.~Finn, E.~Ragoucy, and M.~Vanicat, ``{Integrable boundary
  conditions for multi-species {ASEP}},''
  \href{http://dx.doi.org/10.1088/1751-8113/49/37/375201}{{\em J. Phys. A:
  Math. Theor.} {\bfseries 49} no.~37, (2016) 375201},
  \href{http://arxiv.org/abs/1606.01018}{{\ttfamily arXiv:1606.01018
  [math-ph]}}.

\bibitem{EvansFM09}
M.~Evans, P.~Ferrari, and K.~Mallick, ``{Matrix Representation of the
  Stationary Measure for the multi-species TASEP},''
  \href{http://dx.doi.org/10.1007/s10955-009-9696-2}{{\em J Stat Phys}
  {\bfseries 135} (2009) 217}, \href{http://arxiv.org/abs/0807.0327}{{\ttfamily
  arXiv:0807.0327 [math.PR]}}.

\bibitem{ProlhacEM09}
S.~Prolhac, M.~R. Evans, and K.~Mallick, ``{The matrix product solution of the
  multispecies partially asymmetric exclusion process},''
  \href{http://dx.doi.org/10.1088/1751-8113/42/16/165004}{{\em J. Phys. A:
  Math. Theor.} {\bfseries 42} no.~16, (2009) 165004},
  \href{http://arxiv.org/abs/0812.3293}{{\ttfamily arXiv:0812.3293
  [cond-mat.stat-mech]}}.

\bibitem{AritaAMP11}
C.~Arita, A.~Ayyer, K.~Mallick, and S.~Prolhac, ``{Recursive structures in the
  multispecies TASEP},''
  \href{http://dx.doi.org/10.1088/1751-8113/44/33/335004}{{\em J. Phys. A:
  Math. Theor.} {\bfseries 44} (2011) 335004},
  \href{http://arxiv.org/abs/1104.3752}{{\ttfamily arXiv:1104.3752
  [cond-mat.stat-mech]}}.

\bibitem{AritaAMP12}
C.~Arita, A.~Ayyer, K.~Mallick, and {S. Prolhac}, ``{Generalized matrix Ansatz
  in the multi-species exclusion process - the partially asymmetric case},''
  \href{http://dx.doi.org/10.1088/1751-8113/45/19/195001}{{\em J. Phys. A:
  Math. Gen.} {\bfseries 45} (2012) 195001},
  \href{http://arxiv.org/abs/1201.0388}{{\ttfamily arXiv:1201.0388
  [cond-mat.stat-mech]}}.

\bibitem{KunibaMO15}
A.~Kuniba, S.~Maruyama, and M.~Okado, ``{Multispecies TASEP and combinatorial
  $R^*$},'' \href{http://dx.doi.org/10.1088/1751-8113/48/34/34FT02}{{\em J.
  Phys. A: Math. Theor.} {\bfseries 48} (2015) 34FT02},
  \href{http://arxiv.org/abs/1506.04490}{{\ttfamily arXiv:1506.04490
  [math-ph]}}.

\bibitem{KunibaMO16}
A.~Kuniba, S.~Maruyama, and M.~Okado, ``{Multispecies TASEP and the tetrahedron
  equation},'' \href{http://dx.doi.org/10.1088/1751-8113/49/11/114001}{{\em J.
  Phys. A: Math. Theor.} {\bfseries 49} (2016) 114001},
  \href{http://arxiv.org/abs/1509.09018}{{\ttfamily arXiv:1509.09018
  [nlin.SI]}}.

\bibitem{ZamolodchikovZ79}
A.~B. Zamolodchikov and A.~B. Zamolodchikov, ``{Factorized {$S$}-matrices in
  two dimensions as the exact solutions of certain relativistic quantum field
  theory models},'' \href{http://dx.doi.org/10.1016/0003-4916(79)90391-9}{{\em
  Annals of Physics} {\bfseries 120} no.~2, (1979) 253--291}.

\bibitem{Faddeev80}
L.~D. Faddeev, ``{Quantum completely integrable models in field theory},'' {\em
  Sov. Sci. Rev} {\bfseries C1} (1980) 107.

\bibitem{GhoshalZ94}
S.~Ghoshal and A.~Zamolodchikov, ``{Boundary $S$-matrix and boundary state in
  two-dimensional integrable quantum field theory},''
  \href{http://dx.doi.org/10.1142/S0217751X94001552}{{\em Int. J. Mod. Phys. A}
  {\bfseries 09} (1994) 3841--3885},
  \href{http://arxiv.org/abs/hep-th/9306002}{{\ttfamily arXiv:hep-th/9306002}}.

\bibitem{CrampeRRV16}
N.~Crampe, E.~Ragoucy, V.~Rittenberg, and M.~Vanicat, ``{Integrable dissipative
  exclusion process: Correlation functions and physical properties},''
  \href{http://dx.doi.org/10.1103/PhysRevE.94.032102}{{\em Phys. Rev. E}
  {\bfseries 94} (Sep, 2016) 032102},
  \href{http://arxiv.org/abs/1603.06796}{{\ttfamily arXiv:1603.06796
  [cond-mat.stat-mech]}}.

\bibitem{Sandow94}
S.~Sandow, ``{Partially asymmetric exclusion process with open boundaries},''
  \href{http://dx.doi.org/10.1103/PhysRevE.50.2660}{{\em Phys. Rev. E}
  {\bfseries 50} no.~4, (1994) 2660}.

\bibitem{Lazarescu13}
A.~Lazarescu, {\em {Exact Large Deviations of the Current in the Asymmetric
  Simple Exclusion Process with Open Boundaries}}.
\newblock PhD thesis, Universit\'e Pierre et Marie Curie-Paris 6 and Institut
  de Physique Th\'eorique-CEA Saclay, 2013.
\newblock \href{http://arxiv.org/abs/1311.7370}{{\ttfamily arXiv:1311.7370
  [cond-mat.stat-mech]}}.

\bibitem{Mimachi01}
K.~Mimachi, ``{A duality of {M}acdonald--{K}oornwinder polynomials and its
  applicat ions to integral representations},'' {\em Duke Math. J} {\bfseries
  107} no.~2, (2001) 265--281.

\bibitem{Reshetikhin1992}
N.~Reshetikhin, ``{Jackson-type integrals, bethe vectors, and solutions to a
  difference analog of the Knizhnik-Zamolodchikov system},''
  \href{http://dx.doi.org/10.1007/BF00420749}{{\em Lett Math Phys} {\bfseries
  26} (1992) 153--165}.

\bibitem{ReshetikhinV94}
N.~Reshetikhin and A.~Varchenko, ``{Quasiclassical asymptotics of solutions to
  the KZ equations},'' \href{http://arxiv.org/abs/hep-th/9402126}{{\ttfamily
  arXiv:hep-th/9402126}}.

\end{thebibliography}

\providecommand{\href}[2]{#2}\begingroup\raggedright\endgroup

\end{document}